\documentclass[journal,onecolumn,11pt]{IEEEtran}
\usepackage{subfigure,cite,graphicx,amsmath,amssymb,eufrak,mathrsfs,epsf,epsfig}
\usepackage[usenames]{color}
\usepackage{psfrag}

\usepackage[active]{srcltx}
\ifCLASSINFOpdf
\else
\fi

\begin{document}
\newtheorem{ach}{Achievability}
\newtheorem{con}{Converse}
\newtheorem{definition}{Definition}%[section]
\newtheorem{theorem}{Theorem}%[section]
\newtheorem{lemma}{Lemma}%[section]
\newtheorem{example}{Example}
\newtheorem{cor}{Corollary}%[section]
\newtheorem{prop}{Proposition}%[section]
\newtheorem{conjecture}{Conjecture}%[section]
\newtheorem{remark}{Remark}%[section]
% can use \linebreak \\ within to get better formatting as desired
\title{Degrees of Freedom of Full-Duplex Multiantenna Cellular Networks}
\author{\IEEEauthorblockN{Sang-Woon Jeon,~\IEEEmembership{Member,~IEEE}, Sung Ho Chae,~\IEEEmembership{Member,~IEEE},
\\and Sung Hoon Lim,~\IEEEmembership{Member,~IEEE}}
\thanks{This work has been supported by the Basic Science Research Program through the National Research Foundation of Korea (NRF) funded by the Ministry of Education, Science and Technology (MEST) [NRF-2013R1A1A1064955].}
\thanks{The material in this paper was presented in part at the IEEE Global Communications Conference (GLOBECOM), Austin, TX, December 2014 and has been submitted in part at the IEEE International Symposium on Information Theory (ISIT), Hong Kong, China, June 2015.}
\thanks{S.-W. Jeon is with the Department of Information and Communication Engineering, Andong National University, Andong, South Korea (e-mail: swjeon@anu.ac.kr).}%
\thanks{S. H. Chae, the corresponding author, is with the Digital Media \& Communications (DMC) Research Center, Samsung Electronics, Suwon, South Korea (e-mail: sho.chae00@gmail.com).}%
\thanks{S. H. Lim is with the School of Computer and Communication Sciences,
Ecole Polytechnique F{\'e}d{\'e}rale de Lausanne (EPFL), Lausanne,
Switzerland (e-mail: sung.lim@epfl.ch).}%
}
 \maketitle

\begin{abstract}
We study the \emph{degrees of freedom (DoF)} of cellular networks in which a \emph{full duplex (FD) base station (BS)} equipped with multiple transmit and receive antennas communicates with multiple mobile users. We consider two different scenarios. In the first scenario, we study the case when {\em half duplex (HD) users}, partitioned to either the uplink (UL) set or the downlink (DL) set, simultaneously communicate with the FD BS. In the second scenario, we study the case when {\em FD users} simultaneously communicate UL and DL data with the FD BS. Unlike conventional HD only systems, inter-user interference (within the cell) may severely limit the DoF, and must be carefully taken into account. With the goal of providing theoretical guidelines for designing such FD systems, we completely characterize the sum DoF of each of the two different FD cellular networks by developing an achievable scheme and obtaining a matching upper bound. The key idea of the proposed scheme is to carefully allocate UL and DL information streams using interference alignment and beamforming techniques.
%In particular, the uplink data is sent to the base station using interference alignment such that the inter-user interference is confined within a tolerated number of signal dimensions, while the basestation transmits in the remaining signal dimensions via zero-forcing beamforming for downlink transmission.
By comparing the DoFs of the considered FD systems with those of the conventional HD systems, we establish the DoF gain by enabling FD operation in various configurations. As a consequence of the result, we show that the DoF can approach the two-fold gain over the HD systems when the number of users becomes large enough as compared to the number of antennas at the BS.
\end{abstract}
 \IEEEpeerreviewmaketitle

\section{Introduction}
Current cellular communication systems operate in half-duplex (HD) mode by transmitting and receiving either at different time slots or over different frequency bands. The system is designed in such a way that the downlink (DL) and uplink (UL) traffics are structurally separated by time division duplexing (TDD) or frequency division duplexing (FDD). The advantage of such design principle is that it avoids the high-powered self-interference that is generated during simultaneous transmission and reception. Recent results~\cite{Choi10, Aryafar12,Khandani13, Duarte10, Jain11, Bharadia13}, however, have demonstrated the feasibility of \emph{full-duplex} (FD) wireless communication by suppressing or cancelling self-interference in the RF and baseband level. Various practical designs to realize self-interference cancellation have been proposed in the literature, including adding additional antennas~\cite{Aryafar12}, adding auxiliary transmit RF chains~\cite{Khandani13} or auxiliary receive RF chains~\cite{Duarte10}, using polarization~\cite{Duarte10, Khandani13}, employing balun circuits~\cite{Jain11}, and many more. For more details, see~\cite{Bharadia13, Hong13} and the references therein.

By enabling simultaneous transmission and reception, FD radio is expected to double the spectral efficiency of current HD systems~\cite{Hong13}, and is considered as one of the key technologies for next generation communication systems. Evidently, in situations where the base station (BS) and the user simultaneously transmit bidirectionally as shown in Figure~\ref{fig:first}, enabling FD doubles the overall spectral efficiency. This point-to-point bidirectional communication example, however, is just one instance of how a FD cellular system will function.

In some practical cases, the system may have to support HD users which do not have FD radio due to extra hardware burden on mobile devices. In such case, the FD BS can simultaneously communicate with two sets of  users, one receiving DL data from the BS and the other transmitting UL data to the BS (Figure~\ref{fig:second}).
In another configuration shown in Figure~\ref{fig:third}, for instance, when the BS has many more antennas compared to each user, the FD BS may wish to simultaneously communicate with multiple FD users using multi-user multiple-input and multiple-output (MIMO) techniques. Since the BS is simultaneously transmitting and receiving, there is potential to double the overall spectral efficiency compared to the conventional HD only systems. However, the configurations shown in Figures~\ref{fig:second} and \ref{fig:third} induce a new source of interference that does not arise in HD only networks. In Figure~\ref{fig:second}, since user 1 is transmitting to the BS while user 2 is receiving from the BS, the transmission from user 1 causes interference to user 2. Similarly, in Figure~\ref{fig:third}, the UL transmission of the users causes interference to the DL reception to each other. In cases where this type of interference is strong and proper interference mitigation techniques are not applied, the gain of having FD radios can be severely limited even when self-interference is completely removed.

\begin{figure}[t]
\footnotesize
\begin{center}
\subfigure[Bidirectional full-duplex.]{%
            \label{fig:first}
            \includegraphics[width=0.24\textwidth]{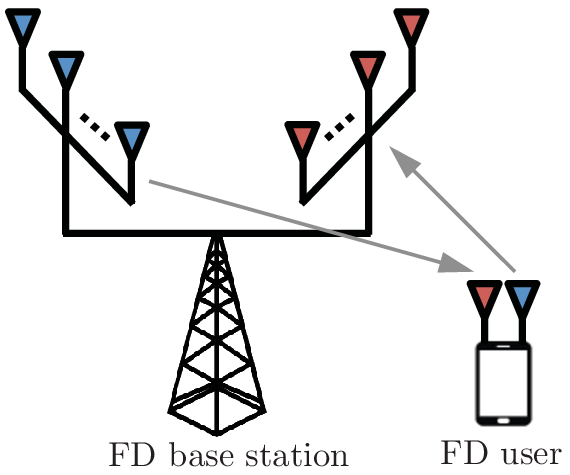}
        }%
\subfigure[Full-duplex at the BS only.]{%
   \label{fig:second}
   \includegraphics[width=0.32\textwidth]{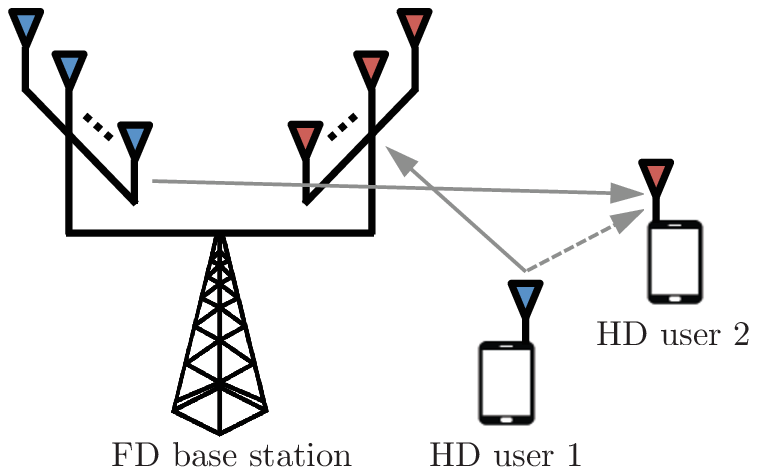}
        }
\subfigure[Full-duplex at both the BS and the users.]{%
   \label{fig:third}
   \includegraphics[width=0.34\textwidth]{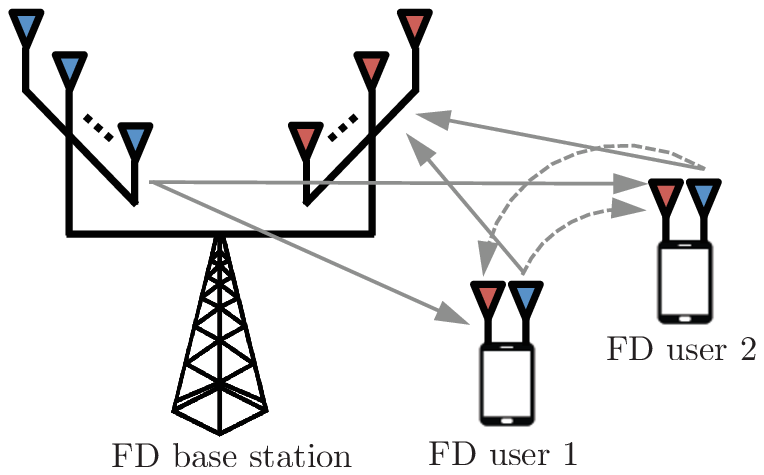}
        }
\end{center}
\caption{Full-duplex network configurations.}
\vspace{-4mm}
\end{figure}

To manage inter-user interference and fully utilize wireless spectrum with FD operation, in this paper we employ signal space \emph{interference alignment} (IA) schemes optimized for FD networks including the cases in  Figure 1. Initially proposed by the seminar works in \cite{Jafar08,Maddah-Ali:08,Cadambe107}, IA is a coding technique that efficiently deals with interference and is known to achieve the optimal DoF for various interference networks~\cite{Suh08,Suh:11,Viveck1:09,Viveck2:09,Tiangao:10,Annapureddy:11,Ke:12,Tiangao:12,Jeon4:12,JeonSuh:14,Sahai13,Nazer11:09}. Especially, it is shown that IA can be successfully applied to mitigate interference in various cellular networks, such as two-cell cellular networks~\cite{Suh08,Suh:11} and multiantenna UL--DL cellular networks~\cite{JeonSuh:14}. Furthermore, the idea of IA can also be applied to the (multi-user) bidirectional cellular network with ergodic phase fading~\cite{Sahai13}, in which the achievable scheme is based on the ergodic IA scheme proposed in \cite{Nazer11:09}.

Motivated by the aforementioned previous works related to IA, we propose the optimal transmission schemes that attain the optimal sum DoFs for two configurations: 1)
a cellular network with a multiantenna FD BS and HD users (Figure~\ref{fig:second}); 2) a cellular network with a multiantenna FD BS and FD users (Figure~\ref{fig:third}).
The key idea of the proposed schemes is to carefully allocate the UL and DL information streams using IA and beamforming techniques. The UL data is sent to the BS using IA such that the inter-user interference is confined within a tolerated number of signal dimensions, while the BS transmits in the remaining signal dimensions via zero-forcing beamforming for the DL transmission.

With the proposed schemes, our primary goal is to answer whether if FD operation can still double the overall spectral efficiency even in the presence of inter-user interference.
We answer this question by providing matching upper bounds with the proposed achievable schemes, completely characterising the sum DoFs of the considered networks.
As a consequence of the result, even in the presence of inter-user interference, we show that the overall DoF can approach the two-fold gain over HD only networks when the number of users becomes large as compared to the number of antennas at the BS. We further provide the DoF gain of the FD systems by considering various configurations (see Sections III and VI.).

\subsection{Previous Works}%%Need to be modified
In~\cite{Cadambe107}, Cadambe and Jafar proposed a novel interference management technique called \emph{interference alignment} (IA), which achieves the optimal sum DoF of $\frac{K}{2}$ for the $K$-user interference channel (IC) with time-varying channel coefficients.
In addition, for the case in which all channel
coefficients are constant, Motahari et al.~\cite{Motahari:09,Motahari:091} proposed a different type of IA scheme
based on number-theoretic properties of rational and
irrational numbers and showed that the optimal DoF of
$\frac{K}{2}$ is also achievable.
Later, alternative methods of aligning interference in the finite signal-to-noise regime has been also proposed in \cite{Nazer11:09,Jeon5:13,Jeon2:11,Jeon2:14}.
The concept of IA has been successfully adapted to various network environments, e.g., see \cite{Viveck2:09,Tiangao:10,Viveck1:09,Annapureddy:11,Ke:12,Tiangao:12,Jeon4:12} and the references therein.

The DoF of cellular networks has been first studied by Suh and Tse for both UL and DL environments, where inter-cell interference exists~\cite{Suh08,Suh:11}. It was shown that, for two-cell networks having $K$ users in each cell,
the sum DoF of $\frac{2K}{K+1}$ is achievable for both UL and DL. Thus, multiple users at each cell are beneficial for
improving the DoF of cellular networks. These models were further extended to more general cases in terms of the number of users and the number of antennas at each BS \cite{Kim:11,Shin:11,Liu2:13,Liu:13,Sridharan:13,Park:12}.
In addition, recently, the DoF of the multiantenna UL--DL cellular network consisting of DL and UL cells has been studied in \cite{JeonSuh:14, KimJeon:14}. For a cellular network with FD operation in the absence of self-interference, the DoF of the (multi-user) bidirectional case has been studied in~\cite{Sahai13} for ergodic phase fading setting. %Furthermore, the achievable DoF of the cellular network with a multiantenna FD BS and HD users (Figure~\ref{fig:second}) was studied in our previous work~\cite{chae14}, in which the sum DoF is now completely characterized in this paper by providing the converse proof and improving the achievable scheme. 

\subsection{Paper Organization}
 The rest of this paper is organized as follows. In Section II, we
describe the network model and the sum DoF metric considered in this paper. In Section
III, we present the main results of the paper and intuitively explain
how FD operation can increase the DoF. In Sections IV and V, we provide
the achievability and converse proofs of the main theorems, respectively. In Section VI, we discuss the impacts of self-interference and scheduling on the DoF. Finally, we conclude in Section VII.

\textbf{Notations}: We will use boldface lowercase letters to denote vectors and boldface uppercase letters to denote matrices.
Throughout the paper, $[1:n]$ denotes $\{1,2,\cdots,n\}$, $\mathbf{0}_n$ denotes the $n\times 1$ all-zero vector, and $\mathbf{I}_n$ denotes the $n\times n$ identity matrix.
For a real value $a$, $a^+$ denotes $\max(0,a)$.
For a set of vectors $\{\mathbf{a}_i\}$, $\operatorname{span}(\{\mathbf{a}_i\})$ denotes the vector space spanned by the vectors in $\{\mathbf{a}_i\}$.
For a vector $\mathbf{b}$, $\mathbf{b}\perp\operatorname{span}(\{\mathbf{a}_i\})$ means that $\mathbf{b}$ is orthogonal with all vectors in $\operatorname{span}(\{\mathbf{a}_i\})$.
%For a matrix $\mathbf{A}$, $\mathbf{A}^{\dagger}$ denotes the transpose of $\mathbf{A}$.
For a set of matrices $\{\mathbf{A}_i\}$, $\operatorname{diag}(\mathbf{A}_1,\cdots, \mathbf{A}_n)$ denotes the block diagonal matrix consisting of $\{\mathbf{A}_i\}$.

\section{Problem Formulation} \label{sec:problem}
For a comprehensive understanding of the DoF improvement by incorporating FD operation, we consider two types of network models: the first network model consists of a single FD BS which simultaneously transmits to a set of DL users (in HD mode) and receives from a set of UL users (in HD mode); the second model consists of a single FD BS communicating with a set of FD users. Unless otherwise specified, we simply denote BS for FD BS in the rest of this paper.

\subsection{Network Model}
In this subsection, we formally define the network models for the two cases mentioned above.

\begin{figure}[t!]
\begin{center}
\includegraphics[scale=0.8]{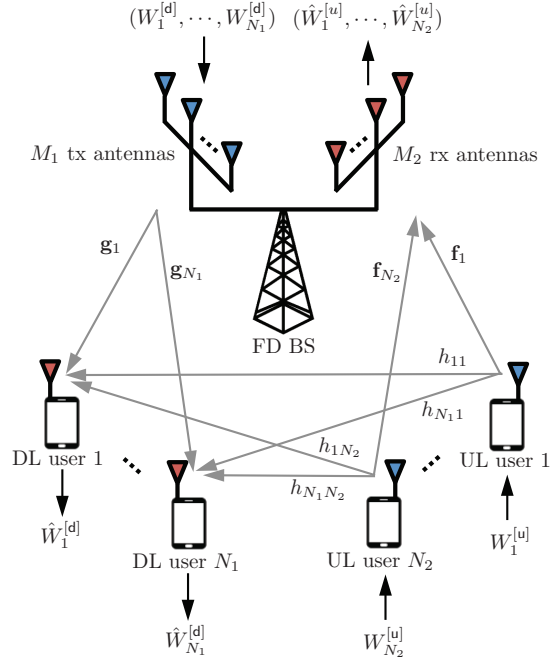}
\end{center}
\vspace{-0.15in}
\caption{The $(M_1,M_2,N_1,N_2)$ FD-BS--HD-user cellular network.}
\label{fig:model1}
\vspace{-0.1in}
\end{figure}

\subsubsection{FD-BS--HD-user cellular networks} \label{sebsec:FD_BS_HD_user}
This network model consists of a mixture of a FD BS and HD users. The HD users are partitioned into two sets, in which one set of users are transmitting to the BS, and the other set of users are receiving from the BS simultaneously. This cellular network is depicted in Figure \ref{fig:model1}. We assume that the FD BS is equipped with $M_1$ transmit antennas and $M_2$ receive antennas. On the user side, we assume that there are $N_1$ DL users and $N_2$ UL users, each equipped with a single antenna.
Here, each user is assumed to operate in HD mode. The BS wishes to send a set of independent messages $(W^{[{\sf d}]}_1,\cdots,W^{[{\sf d}]}_{N_1})$ to the DL users and at the same time wishes to receive a set of independent messages $(W^{[{\sf u}]}_1,\cdots,W^{[{\sf u}]}_{N_2})$ from the UL users.

For $i\in[1:N_1]$, the received signal of DL user $i$ at time $t$, denoted by $y^{[\sf d]}_i(t)\in\mathbb{R}$, is given by
\begin{align} \label{eq:inout1}
y^{[\sf d]}_i(t)=\mathbf{g}_i(t)\mathbf{x}^{[\sf bs]}(t)+\sum_{j=1}^{N_2}h_{ij}(t)x^{[\sf u]}_j(t)+z_i^{[\sf d]}(t)
\end{align}
and the received signal vector of the BS at time $t$, denoted by $\mathbf{y}^{[\sf bs]}(t)\in\mathbb{R}^{M_2\times 1}$, is given by
\begin{align} \label{eq:inout2}
\mathbf{y}^{[\sf bs]}(t)=\sum_{j=1}^{N_2}\mathbf{f}_j(t)x_j^{[\sf u]}(t)+\mathbf{z}^{[\sf bs]}(t),
\end{align}
where $\mathbf{x}^{[{\sf bs}]}(t)\in\mathbb{R}^{M_1\times1}$ is the transmit signal vector of the BS at time $t$, $x^{[{\sf u}]}_j[t]\in\mathbb{R}$ is the transmit signal of UL user $j$ at time $t$, $\mathbf{g}_i(t)\in \mathbb{R}^{1\times M_1}$ is the channel vector from the BS to DL user $i$ at time $t$, $h_{ij}(t)\in \mathbb{R}$ is the scalar channel from UL user $j$ to DL user $i$ at time $t$, and $\mathbf{f}(t)\in \mathbb{R}^{M_2\times 1}$ is the channel vector from UL user $j$ to the BS.
The additive noises $z_i^{[{\sf d}]}(t)\in\mathbb{R}$ and $\mathbf{z}^{[{\sf bs}]}(t)\in\mathbb{R}^{M_2\times 1}$ are assumed to be independent of each other and also independent over time, and is distributed as $z_i^{[{\sf d}]}(t)\sim \mathcal{N}(0,1)$ and $\mathbf{z}^{[{\sf bs}]}(t)\sim \mathcal{N}(\mathbf{0}_{M_2},\mathbf{I}_{M_2})$.

We assume that channel coefficients are drawn i.i.d. from a continuous distribution and vary independently over time.
It is further assumed that global channel state information (CSI) is available at the BS and each UL and DL user.
The BS and each UL user is assumed to satisfy an average transmit power constraint, i.e., ${\sf E}\left[\|\mathbf{x}^{[{\sf bs}]}(t)\|^2\right]\leq P$ and ${\sf E}\big[|x_j^{[{\sf u}]}(t)|^2\big]\leq P$ for all $j\in[1:N_2]$.

In the rest of the paper, we denote this network as a $(M_1,M_2,N_1,N_2)$ FD-BS--HD-user cellular network.
\medskip
\begin{remark}
We assume perfect self-interference suppression within the BS during FD operation. Hence there is no self-interference for the input--output relations in \eqref{eq:inout1} and \eqref{eq:inout2}. We will discuss how imperfect self-interference suppression effects the DoF in Section VI-A. \hfill$\lozenge$
\end{remark}

\begin{figure}[t!]
\begin{center}
\includegraphics[scale=0.8]{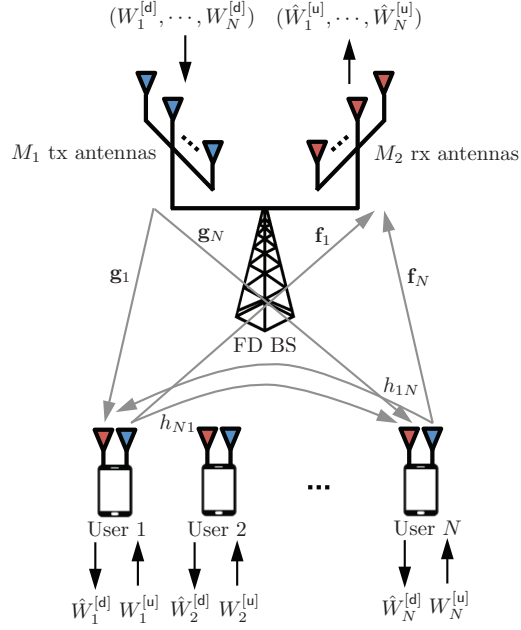}
\end{center}
\vspace{-0.15in}
\caption{The $(M_1,M_2,N)$ FD-BS--FD-user cellular network.}
\label{fig:model2}
\vspace{-0.1in}
\end{figure}

\subsubsection{FD-BS--FD-user cellular networks} \label{sebsec:FD_BS_FD_user}
In this model, we consider the case where both the BS and users have FD capability (depicted in Figure \ref{fig:model2}). As before, we assume that the BS is equipped with $M_1$ transmit antennas and $M_2$ receive antennas. However, unlike the FD-BS--HD-user cellular network, there is a single set of $N$ FD users, each equipped with a single transmit and a single receive antenna, that simultaneously transmits to and receives from the BS.
The BS wishes to send a set of independent messages $(W^{[{\sf d}]}_1,\cdots,W^{[{\sf d}]}_{N})$ to the users and at the same time wishes to receive a set of independent messages $(W^{[{\sf u}]}_1,\cdots,W^{[{\sf u}]}_{N})$ from the same users.

For $i\in[1,N]$, the received signal of user $i$ at time $t$ is given by
\begin{align} \label{eq:inout3}
y_i(t)=\mathbf{g}_i(t)\mathbf{x}^{[\sf bs]}(t)+\sum_{j=1,j\neq i}^{N}h_{ij}(t)x_j(t)+z_i(t)
\end{align}
and the received signal vector of the BS at time $t$ is given by
\begin{align} \label{eq:inout4}
\mathbf{y}^{[\sf bs]}(t)=\sum_{j=1}^{N}\mathbf{f}_j(t)x_j(t)+\mathbf{z}^{[\sf bs]}(t).
\end{align}
As before, we assume that self-interference at the BS and each user is completely suppressed, which is reflected in the input--output relations in \eqref{eq:inout3} and \eqref{eq:inout4}. The rest of the assumptions are the same as those of the $(M_1,M_2,N_1,N_2)$ FD-BS--HD-user cellular network.

In the rest of the paper, we denote this network as a $(M_1,M_2,N)$ FD-BS--FD-user cellular network.

\subsection{Degrees of Freedom} \label{subsec:dof_definition}
For each network model, we define a set of length $n$ block codes and its achievable DoF.

\subsubsection{FD-BS--HD-user cellular networks} \label{sebsec:dof_FD_BS_HD_user}
Let $W^{[{\sf d}]}_i$ and $W^{[{\sf u}]}_j$ be chosen uniformly at random from $[1:2^{nR^{[{\sf d}]}_i}]$ and $[1:2^{nR^{[{\sf u}]}_j}]$ respectively, where $i\in[1:N_1]$ and $j\in[1:N_2]$.
Then a $(2^{nR^{[{\sf d}]}_1},\cdots,2^{nR^{[{\sf d}]}_{N_1}},2^{nR^{[{\sf u}]}_1},\cdots,2^{nR^{[{\sf u}]}_{N_2}};n)$ code consists of the following set of encoding and decoding functions:
\begin{itemize}
\item {\em Encoding:} For $t\in[1:n]$, the encoding function of the BS at time $t$ is given by 
\begin{align*}
\mathbf{x}^{[\sf bs]}(t)=\phi_t(W_1^{[{\sf d}]},\cdots,W_{N_1}^{[{\sf d}]},\mathbf{y}^{[{\sf bs}]}(1),\cdots \mathbf{y}^{[{\sf bs}]}(t-1)).
\end{align*}
For $t\in[1:n]$, the encoding function of UL user $j$ at time $t$ is given by 
\begin{align*}
x_j(t)=\varphi_t(W_j^{[{\sf u}]}),
\end{align*}
where $j\in[1:N_2]$.
\item {\em Decoding:} Upon receiving $\mathbf{y}^{[{\sf bs}]}(1)$ to $\mathbf{y}^{[{\sf bs}]}(n)$, the decoding function of the BS is given by 
\begin{align*}
	\hat{W}^{[{\sf u}]}_j=\chi_j(\mathbf{y}^{[{\sf bs}]}(1),\cdots,\mathbf{y}^{[{\sf bs}]}(n),W_1^{[{\sf d}]},\cdots,W_{N_1}^{[{\sf d}]}) \text{ for } j\in[1:N_2]. 
\end{align*}
Upon receiving $y_i(1)$ to $y_i(n)$, the decoding function of DL user $i$ is given by 
\begin{align*}
\hat{W}^{[{\sf d}]}_i=\psi_i(y_i(1),\cdots,y_i(n)),
\end{align*} 
where $i\in[1:N_1]$.
\end{itemize}

A rate tuple $(R^{[{\sf d}]}_1,\cdots,R^{[{\sf d}]}_{N_1},R^{[{\sf u}]}_1,\cdots,R^{[{\sf u}]}_{N_2})$ is said to be \emph{achievable} for the FD-BS--HD-user cellular network if there exists a sequence of  $(2^{nR^{[{\sf d}]}_1},\cdots,2^{nR^{[{\sf d}]}_{N_1}},2^{nR^{[{\sf u}]}_1},\cdots,2^{nR^{[{\sf u}]}_{N_2}};n)$ codes such that $\Pr(\hat{W}^{[{\sf d}]}_i\neq W^{[{\sf d}]}_i)\to 0$ and $\Pr(\hat{W}^{[{\sf u}]}_j\neq W^{[{\sf u}]}_j)\to 0$ as $n$ increases for all $i\in[1:N_1]$ and $j\in[1:N_2]$.
Then the achievable DoF tuple is given by
\begin{align}
(d^{[{\sf d}]}_1,\cdots,d_{N_1}^{[{\sf d}]},d^{[{\sf u}]}_1,\cdots,d_{N_2}^{[{\sf u}]})=\lim_{P\to\infty}\left(\frac{R^{[{\sf d}]}_1}{\frac{1}{2}\log P},\cdots,\frac{R^{[{\sf d}]}_{N_1}}{\frac{1}{2}\log P},\frac{R^{[{\sf u}]}_1}{\frac{1}{2}\log P},\cdots,\frac{R^{[{\sf u}]}_{N_2}}{\frac{1}{2}\log P}\right).
\end{align}
We further denote the maximum achievable sum DoF of the FD-BS--HD-user cellular network by $d_{\Sigma,1}$, i.e.,
\begin{align}
d_{\Sigma,1}=\max_{(d^{[{\sf d}]}_1,\cdots,d_{N_1}^{[{\sf d}]},d^{[{\sf u}]}_1,\cdots,d_{N_2}^{[{\sf u}]})\in\mathcal{D}}\left\{\sum_{i=1}^{N_1}d_i^{[{\sf d}]}+\sum_{j=1}^{N_2}d_j^{[{\sf u}]}\right\},
\end{align}
where $\mathcal{D}$ denotes the DoF region of the FD-BS--HD-user cellular network.

\subsubsection{FD-BS--FD-user cellular networks} \label{sebsec:dof_FD_BS_FD_user}

Similar to the FD-BS--HD-user cellular network, we can define an achievable DoF tuple of the FD-BS--FD-user cellular network.
The key difference is that each user also operates in FD mode for this second model.
Specifically, the encoding function of user $i$ at time $t\in[1:n]$ is given by $x_i(t)=\varphi_t(W_i^{[{\sf u}]},y_i(1),\cdots,y_i(t-1))$ and the decoding function of user $i$ is given by $\hat{W}^{[{\sf d}]}_i=\psi_i(y_i(1),\cdots,y_i(n),W^{[{\sf u}]}_i)$, where $i\in[1:N]$.
Then the definition of an achievable DoF tuple $(d^{[{\sf d}]}_1,\cdots,d_{N}^{[{\sf d}]},d^{[{\sf u}]}_1,\cdots,d_{N}^{[{\sf u}]})$ is the same as that of the FD-BS--HD-user cellular network. Similarly, we denote the maximum achievable sum DoF of the FD-BS--FD-user cellular network by $d_{\Sigma,2}$.

\section{Main Results}

In this section, we state the main results of this paper. We completely characterize the sum DoFs of both the $(M_1,M_2,N_1,N_2)$ FD-BS--HD-user cellular network and the $(M_1,M_2,N)$ FD-BS--FD-user cellular network.

\begin{theorem} \label{thm:dof_1}
For the $(M_1,M_2,N_1,N_2)$ FD-BS--HD-user cellular network,
\begin{align} \label{eq:dof_1}
d_{\Sigma,1}=\min\left\{M_1+M_2,\max(N_1,N_2),\max\left(M_1+\frac{N_2(N_1-M_1)}{N_1},M_2+\frac{N_1(N_2-M_2)}{N_2}\right)\right\}.
\end{align}
\end{theorem}
\begin{proof}
The achievability proof is given in Section \ref{sec:achievability} and the converse proof is given in Section \ref{sec:converse}.
\end{proof}

We demonstrate the utility of Theorem~\ref{thm:dof_1} by the following example.
\begin{example}[Symmetric FD-BS--HD-user cellular networks] \label{ex:symmetric1}
Consider the $(M,M,N,N)$ FD-BS--HD-user cellular network, i.e., $M_1=M_2=M$ and $N_1=N_2=N$. For this symmetric case, $d_{\Sigma,1}=\min(2M,N)$ from Theorem \ref{thm:dof_1}. On the other hand, if the BS operates in HD mode, we can easily see that the sum DoF is limited by $\min(M,N)$. By comparing the sum DoFs, we can see that there is a two-fold gain by operating the BS in FD mode when we have enough number of users in the network, i.e., $N\ge 2M$. Figure \ref{Ex1} plots $d_{\Sigma,1}$ with respect to $N$ when $M=5$. As shown in the figure, FD operation at the BS improves the sum DoF as $N$ increases and eventually the sum DoF is doubled compared to HD BS for large enough $N$. \hfill$\lozenge$
\end{example}

\begin{figure}[!t]
\centering
\includegraphics[scale=0.7]{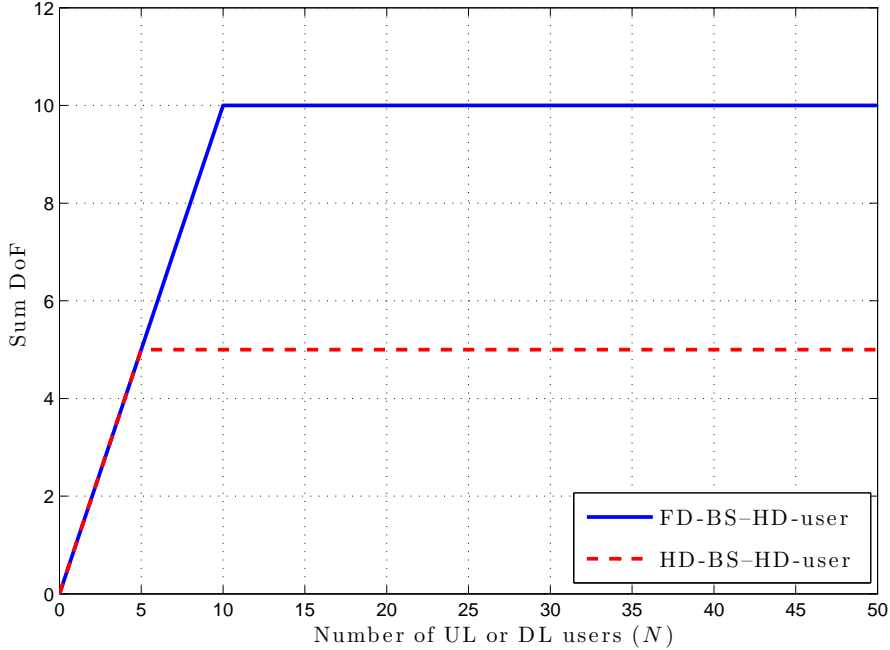}
\caption{Sum DoFs when $M_1=M_2=5$ and $N_1=N_2=N$.} \label{Ex1}
\end{figure}

For the FD-BS--FD-user cellular network, we have the following theorem.

\begin{theorem} \label{thm:dof_2}
For the $(M_1,M_2,N)$ FD-BS--FD-user cellular network,
\begin{align} \label{eq:dof_2}
d_{\Sigma,2}=\min(M_1+M_2,N).
\end{align}
\end{theorem}
\begin{proof}
From the network model and the DoF definition in Section \ref{sec:problem}, any achievable sum DoF in the $(M_1,M_2,N,N)$ FD-BS--HD-user cellular network is also achievable for the $(M_1,M_2,N)$ FD-BS--FD-user cellular network. In particular, the encoding functions at the BS are the same for both network models, and the BS also receives the same signal as shown in \eqref{eq:inout1} and \eqref{eq:inout3}. Comparing the user encoders, we can see that the user encoding function for the FD-BS--FD-user cellular network is more general than the encoding function for the FD-BS--HD-user cellular network. Furthermore, we can easily see that the received signal \eqref{eq:inout4} is ``better'' than the received signal for the  FD-BS--HD-user cellular network \eqref{eq:inout3}, in that it has less interference (self-interference is suppressed for the FD user case). Hence, from Theorem \ref{thm:dof_1}, the sum DoF of $\min(M_1+M_2,N)$ is achievable for the $(M_1,M_2,N)$ FD-BS--FD-user cellular network, which coincides with $d_{\Sigma,2}$ in \eqref{eq:dof_2}.
The converse proof is given in Section \ref{sec:converse}.
\end{proof}

%\begin{remark}[Gain from FD operation at users] Comparing the result of Theorem 1 and Theorem 2, we can see that by additionally having FD capability at the user side, we can achieve the same DoF gain as if inter-user interference is completely removed. For a fair comparison, let $N$ denote the number of total users in a FD-BS--HD-user cellular network, i.e., $N=N_1+N_2$. Then, it is easy to see that if there is no inter-user interference, the sum DoF of FD-BS--HD-user cellular network becomes $\min(M_1,N_1)+\min(M_2,N_2)=\min(M_1+M_2,N_1+N_2)=\min(M_1+M_2,N)$, which coincides with the sum DoF in Theorem 2. \hfill$\lozenge$
%\end{remark}

%Let $N$ denote the number of total users in a FD-BS--HD-user cellular network, i.e., $N=N_1+N_2$. Then, it is easy to see that if there is no inter-user interference, the sum DoF of FD-BS--HD-user cellular network becomes $\min(M_1,N_1)+\min(M_2,N_2)=\min(M_1+M_2,N_1+N_2)=\min(M_1+M_2,N)$, which coincides with the sum DoF in Theorem 2. Therefore, we can see that allowing FD operation at users can effectively remove inter-user interference in FD-BS--HD-user cellular network in terms of the sum DoF.

We demonstrate the utility of Theorem~\ref{thm:dof_2} by the following example.

\begin{example}[Symmetric FD-BS--FD-user cellular networks] \label{ex:symmetric2}
Consider the $(M,M,N)$ FD-BS--FD-user cellular network, i.e., $M_1=M_2=M$. For this symmetric case, $d_{\Sigma,2}=\min(2M,N)$ from Theorem \ref{thm:dof_2}, which coincides with the sum DoF of the symmetric FD-BS--HD-user cellular network in Example \ref{ex:symmetric1}. Again, if both the BS and the users are limited to operate in HD mode, then the sum DoF is limited by $\min(M,N)$. \hfill$\lozenge$
\end{example}
\medskip

To be fair, the $(M,M,N)$ FD-BS--FD-user cellular network in Example \ref{ex:symmetric2} has been considered in \cite{Sahai13} under the ergodic fading setting assuming that the phase of each channel coefficient in $\{h_{ij}(t)\}_{i,j\in[1:N],i\neq j}$ is drawn independently from a uniform phase distribution. For this case, it has been shown in \cite[Theorem 1]{Sahai13} that the achievable DoF tuple satisfies:
\begin{align} \label{eq:dof_region}
\sum_{i=1}^{N}d^{[{\sf d}]}_i&\leq \min(M,N)\nonumber\\
\sum_{j=1}^{N}d^{[{\sf u}]}_j&\leq \min(M,N)\nonumber\\
\sum_{i=1}^{N}d^{[{\sf d}]}_i+\sum_{j=1}^{N}d^{[{\sf u}]}_j&\leq\min(2M,N),
\end{align}
where \eqref{eq:dof_region} characterises the sum DoF. This result in \cite{Sahai13} is general in that it provides a general achievable {\em DoF region}, while our result in Theorem~\ref{thm:dof_2} generalizes the {\em sum DoF} result in \cite{Sahai13} by considering arbitrary number of transmit and receive antennas at the BS, and also extends to any i.i.d. generic channel setting including the ergodic fading setting.

In Section \ref{sec:discussion}, we discuss in detail regarding the DoF improvement by enabling FD operation, and also the effect of imperfect self-interference suppression.

\section{Achievability} \label{sec:achievability}
In this section, we prove that the sum DoF $d_{\Sigma,1}$ in Theorem \ref{thm:dof_1} is achievable.
To better illustrate the main insight of the coding scheme, we first consider the achievablity of Theorem \ref{thm:dof_1} for the case $N_1=1$ in Section \ref{subsec:case1}. The main component of the scheme utilizes IA via transmit beamforming with a finite symbol extension. For general $N_1$, interference from multiple UL users should be simultaneously aligned at multiple DL users, which requires asymptotic IA, i.e., an arbitrarily large symbol extension.
In Section \ref{subsec:case2}, we introduce transmit beamforming adopting such asymptotic IA for the general network configuration.

\begin{figure}[t!]
\begin{center}
\includegraphics[scale=0.8]{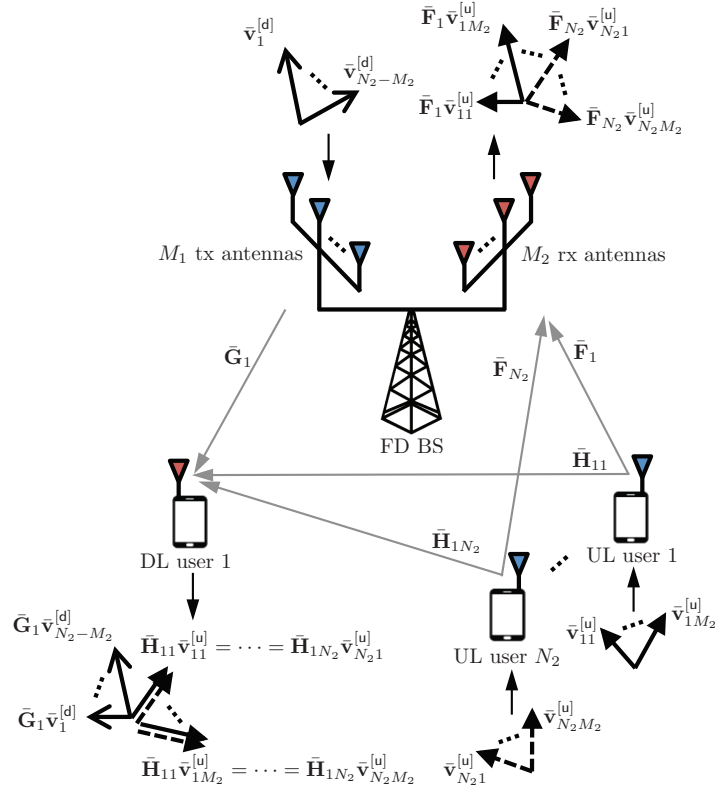}
\end{center}
\vspace{-0.15in}
\caption{Transmit beamforming for the $(M_1,M_2,1,N_2)$ FD-BS--HD-user cellular network when $M_2\leq N_2$.}
\label{fig:scheme1_model1}
\vspace{-0.1in}
\end{figure}

\subsection{The Case $N_1=1$} \label{subsec:case1}
For the $(M_1,M_2,1,N_2)$ FD-BS--HD-user cellular network,
\begin{align} \label{eq:dof_N_1}
d_{\Sigma,1}=\begin{cases}
N_2 &\text{if }M_2\geq N_2,\\
M_2+\frac{N_2-M_2}{N_2} &\text{if }M_2\leq N_2
\end{cases}
\end{align}
from Theorem \ref{thm:dof_1}. For the proof on how \eqref{eq:dof_N_1} can be evaluated from \eqref{eq:dof_1} for the case $N_1=1$, we refer to the proof in Lemma \ref{lemma:sum_DoF_conversion}.
In the following, we show that $d_{\Sigma,1}$ in \eqref{eq:dof_N_1} is achievable by considering two cases, $M_2\ge N_2$ and $M_2\le N_2$. For the first case $M_2\geq N_2$, we can easily achieve $d_{\Sigma,1}=N_2$ by simply utilizing only the UL transmission, i.e., the BS receives from the $N_2$ UL users with $M_2$ receive antennas.
Now consider the second case where $M_2\leq N_2$, which we explain with the help of Figure \ref{fig:scheme1_model1}.

For this case, communication takes place via transmit beamforming over a block of $N_2$ time slots, i.e., $N_2$ symbol extension.
Denote
\begin{align}
\bar{\mathbf{G}}_1&=\operatorname{diag}(\mathbf{g}_1(1),\cdots\mathbf{g}_1(N_2))\in\mathbb{R}^{N_2\times M_1N_2},\nonumber\\
 \bar{\mathbf{H}}_{1j}&=\operatorname{diag}(h_{1j}(1),\cdots,h_{1j}(N_2))\in \mathbb{R}^{N_2\times N_2},
\nonumber\\
\bar{\mathbf{F}}_j&=\operatorname{diag}(\mathbf{f}_j(1),\cdots,\mathbf{f}_j(N_2))\in\mathbb{R}^{M_2N_2\times N_2},
\end{align}
  where $j\in[1:N_2]$.
The BS sends $N_2-M_2$ information symbols to the DL user via the $M_1N_2\times 1$ beamforming vectors $\{\bar{\mathbf{v}}^{[{\sf d}]}_k\}_{k\in[1:N_2-M_2]}$. On the other hand, UL user $j\in[1:N_2]$ sends $M_2$ information symbols to the BS via the $N_2\times 1$ beamforming vectors $\{\bar{\mathbf{v}}^{[{\sf u}]}_{jk}\}_{k\in[1:M_2]}$.

We first construct $\{\bar{\mathbf{v}}^{[{\sf d}]}_k\}_{k\in[1:N_2-M_2]}$ as a set of $N_2-M_2$ linearly independent random vectors.
Next, we construct linearly independent $\{\bar{\mathbf{v}}^{[{\sf u}]}_{jk}\}_{k\in[1:M_2],j\in[1:N_2]}$ such that for each $k\in[1:M_2]$, all the $N_2$ information symbols that are indexed with $k\in[1:M_2]$ are aligned at the DL user, i.e., satisfying the IA condition $\bar{\mathbf{H}}_{11}\bar{\mathbf{v}}^{[{\sf u}]}_{1k}=\cdots=\bar{\mathbf{H}}_{1N_2}\bar{\mathbf{v}}^{[{\sf u}]}_{N_2k}$ for all $k\in[1:M_2]$.
Specifically, we first construct $\{\bar{\mathbf{v}}^{[{\sf u}]}_{1k}\}_{k\in[1:M_2]}$ as a set of $M_2$ linearly independent random vectors. Then, for a given $\{\bar{\mathbf{v}}^{[{\sf u}]}_{1k}\}_{k\in[1:M_2]}$, we construct $\bar{\mathbf{v}}^{[{\sf u}]}_{jk}=(\bar{\mathbf{H}}_{1j})^{-1}\bar{\mathbf{H}}_{11}\bar{\mathbf{v}}^{[{\sf u}]}_{1k}$ for all $k\in[1:M_2],j\in[2:N_2]$. By such construction, the resulting $\{\bar{\mathbf{v}}^{[{\sf u}]}_{jk}\}_{k\in[1:M_2],j\in[1:N_2]}$ are linearly independent almost surely.

We now move on to the decoding step at the DL user. Due to the previous IA procedure of the UL users, the number of dimensions occupied by the inter-user interference signals is given by $M_2$. Furthermore, the DL signals sent by the BS occupy $N_2-M_2$ dimensions and are linearly independent of the inter-user interference signals almost surely. Hence, the DL user is able to decode its intended information symbols achieving one DoF each.
Next, consider decoding at the BS. Since $\{\bar{\mathbf{v}}^{[{\sf u}]}_{jk}\}_{k\in[1:M_2],j\in[1:N_2]}$ are linearly independent, $\{\bar{\mathbf{F}}_{j}\bar{\mathbf{v}}^{[{\sf u}]}_{jk}\}_{k\in[1:M_2],j\in[1:N_2]}$ are also linearly independent almost surely. Hence, the BS is able to decode the $M_2N_2$ information symbols.
Finally, from the fact that a total of $N_2-M_2+M_2N_2$ information symbols are communicated over $N_2$ time slots, $d_{\Sigma,1}=M_2+\frac{N_2-M_2}{N_2}$ is achievable for the case $M_2\leq N_2$.

%In summary, the BS reserves $M_2$ dimensions out of the total $N_2$ dimensions by only transmitting through $N_2-M_2$ dimensions in the DL transmission. The $M_2$ dimension is reserved for inter-user interference by the UL users. Thus, the UL users objective is to maximize their transmission to the BS, while confining the inter-user interference within the reserved $M_2$ dimensions. For this purpose, the UL users ``pack'' the total of $M_2N_2$ information symbols into $M_2$ dimensions via IA. In the next section, we generalise this approach to general $N_1$.

\begin{figure}[t!]
\begin{center}
\includegraphics[scale=0.8]{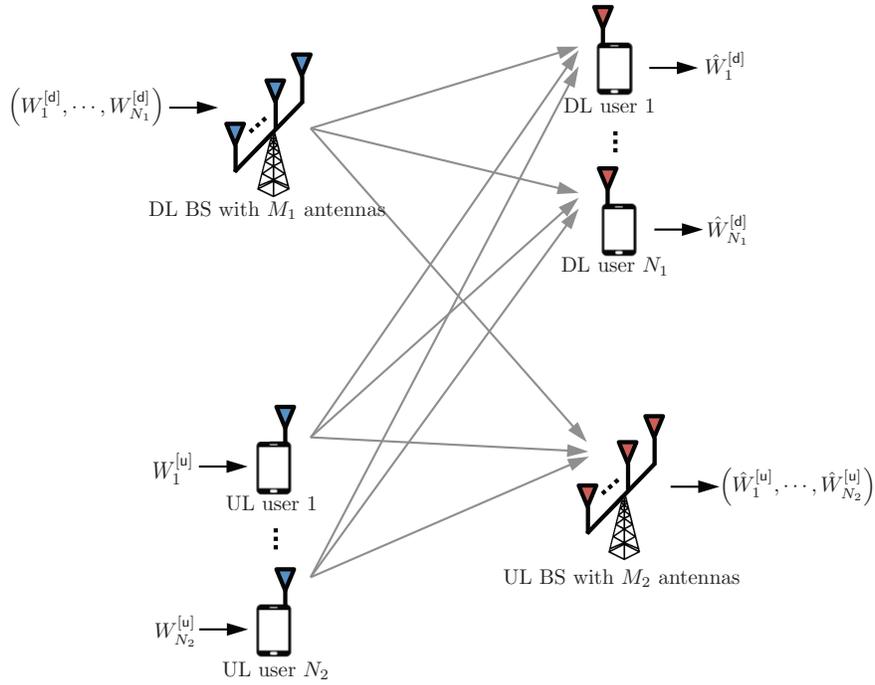}
\end{center}
\vspace{-0.15in}
\caption{Two-cell multiantenna cellular networks in which the first and second cells operate as DL and UL respectively.}
\label{fig:ULDL_model}
\vspace{-0.1in}
\end{figure}

\begin{figure}[t!]
\begin{center}
\includegraphics[scale=0.8]{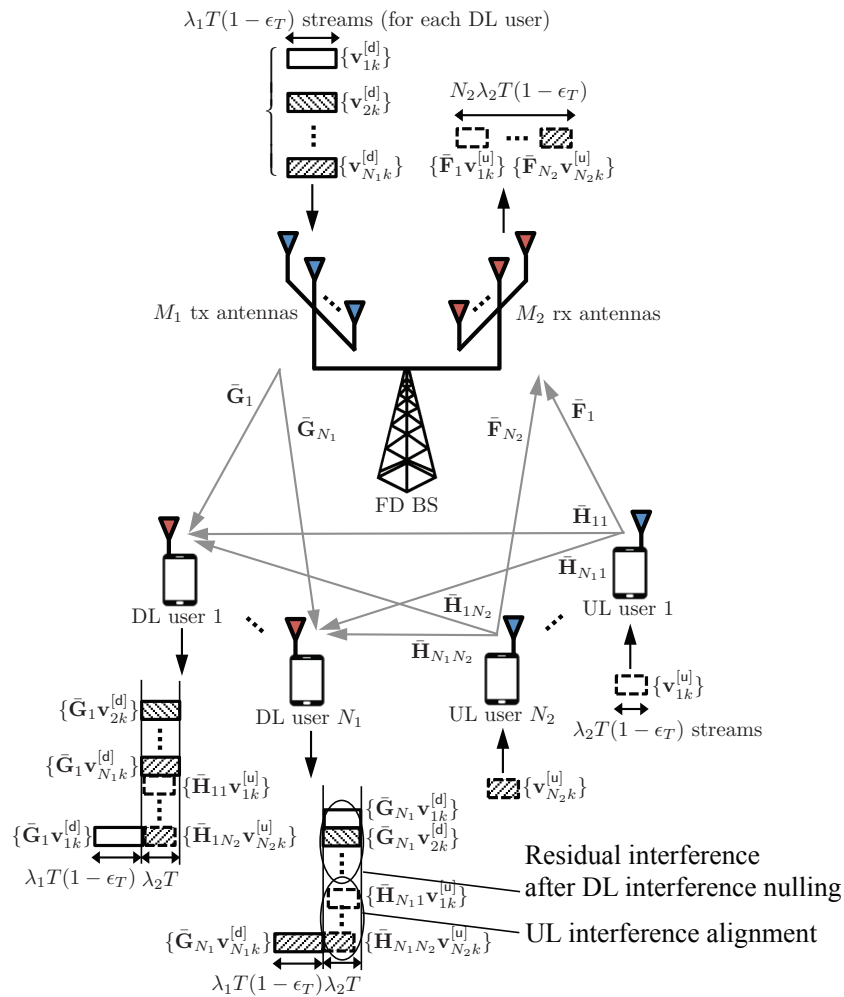}
\end{center}
\vspace{-0.15in}
\caption{Conceptual illustration of transmit beamforming for the $(M_1,M_2,N_1,N_2)$ FD-BS--HD-user cellular network, where for convenience we assume $\lambda_1\geq \lambda_2$ in the figure.}
\label{fig:scheme1_model1}
\vspace{-0.1in}
\end{figure}

\subsection{General Case} \label{subsec:case2}
Following the intuition in the previous subsection, with IA, we would like to confine the interference signals transmitted from multiple UL users into a preserved signal subspace at each DL user, leaving the rest of subspace for the intended signals sent from the BS. For general $N_1$, this requires arbitrarily large number of symbol extensions \cite{Cadambe107}.

For this purpose, a recently developed IA technique in \cite{JeonSuh:14} for the multiantenna UL--DL cellular network can be applied for the $(M_1,M_2,N_1,N_2)$ FD-BS--HD-user cellular network. To show how the scheme in~\cite{JeonSuh:14} fits into our problem, we begin with a brief overview of their network model.
In \cite{JeonSuh:14}, the authors consider a UL--DL cellular network (Figure~\ref{fig:ULDL_model}), where two cells co-exist (each cell consists of one BS and a set of users). In one cell, a BS with $M_1$ antennas transmits to a set of $N_1$ DL users, while in the other cell a set of UL users transmit to a BS with $M_2$ antennas. Thus, the network models the case when it can schedule each cell in DL or UL phase separately. The structural similarity with our FD-BS--HD-user cellular network is apparent, and the key difference between them is that there is no inter-cell interference between the DL BS and UL BS (since in the FD-BS--HD-user cellular network, UL and DL is performed with a single FD BS). Accordingly, the transmit signal vector of the DL BS in the UL--DL model (Figure \ref{fig:ULDL_model}) can also be used as the transmit signal vector of the FD BS in the FD-BS--HD-user cellular network (Figure \ref{fig:model1}), and the transmit signal of each UL user in the UL--DL model (Figure \ref{fig:ULDL_model}) can also be used by each UL user in the FD-BS--HD-user cellular network (Figure \ref{fig:model1}).
Therefore, the IA scheme stated in \cite[Section IV-E]{JeonSuh:14} is applicable to the $(M_1,M_2,N_1,N_2)$ FD-BS--HD-user cellular network.
However, due to the self-interference suppression capability in the FD BS case, the {\em performance} resulting from this scheme will be different for the two networks, and our contribution for achievability lies in the analysis of the sum DoF of the scheme for the FD-BS--HD-user cellular network. %This is due to the difference in the received signal vector at the UL BS; while the UL--DL network is subject to inter-cell interference from the DL BS, the FD BS for the FD-BS--HD-user cellular network can simply cancel such interference due to the self-interference cancellation capability.

For completeness and better understanding, we briefly summarize how the IA scheme in \cite[Section IV-E]{JeonSuh:14} can be adapted to the $(M_1,M_2,N_1,N_2)$ FD-BS--HD-user cellular network. We then give the analysis of its achievable sum DoF.

\subsubsection{DL interference nulling and UL interference alignment}
Communication takes place over a block of $T$ time slots, i.e., $T$ symbol extension. Denote
\begin{align}
\bar{\mathbf{G}}_i&=\operatorname{diag}(\mathbf{g}_1(i),\cdots\mathbf{g}_i(T))\in\mathbb{R}^{T\times M_1T},\nonumber\\
 \bar{\mathbf{H}}_{ij}&=\operatorname{diag}(h_{ij}(1),\cdots,h_{ij}(T))\in \mathbb{R}^{T\times T},
\nonumber\\
\bar{\mathbf{F}}_j&=\operatorname{diag}(\mathbf{f}_j(1),\cdots,\mathbf{f}_j(T))\in\mathbb{R}^{M_2T\times T},
\end{align}
for $i\in[1:N_1]$ and $j\in[1:N_2]$.
Each information symbol is transmitted through a length-$T$ time-extended beamforming vector.
Figure \ref{fig:scheme1_model1} is a conceptual illustration for this transmit beamforming.
We refer to \cite[Section IV-E]{JeonSuh:14} for the detailed construction of beamforming vectors.
Suppose that $\lambda_1,\lambda_2\in(0,1]$ and $\epsilon_T\to 0$ as $T$ increases.
For $i\in[1:N_1]$, the BS sends $\lambda_1T(1-\epsilon_T)$ information symbols to DL user $i$ using the set of $T$ time-extended beamforming vectors $\{\mathbf{v}_{ik}^{[{\sf d}]}\}_{k\in[1:\lambda_1T(1-\epsilon_T)]}$.
Similarly, UL user $j$ sends  $\lambda_2T(1-\epsilon_T)$ information symbols to the BS using the set of $T$ time-extended beamforming vectors $\{\mathbf{v}_{jk}^{[{\sf u}]}\}_{k\in[1:\lambda_2T(1-\epsilon_T)]}$, where $j\in[1:N_2]$.

As seen in Figure~\ref{fig:scheme1_model1}, the set of beamforming vectors transmitted from each UL user is set to align its interference at each DL user.
More specifically, by applying asymptotic IA for $\big\{\bar{\mathbf{v}}^{[{\sf u}]}_{jk}\big\}_{j\in[1:N_2],k\in[1:\lambda_2 T(1-\epsilon_T)]}$, we can guarantee that  $\operatorname{span}\left(\big\{ \bar{\mathbf{H}}_{ij}\bar{\mathbf{v}}^{[{\sf u}]}_{jk}\big\}_{j\in[1:N_2],k\in[1:\lambda_2 T(1-\epsilon_T)]}\right)$ occupies at most $\lambda_2 T$ dimensional subspace in $T$ dimensional signal space for all $i\in[1:N_1]$ almost surely in the limit of large $T$, where $\epsilon_T\to 0$ as $T$ increases, see also \cite[Lemma 2]{JeonSuh:14}.
Then the set of beamforming vectors transmitted from the BS is set to null out its interference at each DL user.
More specifically, $\{\bar{\mathbf{v}}^{[{\sf d}]}_{ik}\}_{i\in[1:N_1],k\in[1:\lambda_1T(1-\epsilon_T)]}$ is set to satisfy
$\bar{\mathbf{G}}_i\bar{\mathbf{v}}^{[{\sf d}]}_{jk}\perp\operatorname{span}\left(\big\{\bar{\mathbf{G}}_{i}\bar{\mathbf{v}}^{[{\sf d}]}_{ik'}\big\}_{k'\in[1:\lambda_1 T(1-\epsilon_T)]}\right)$ for all $i,j\in[1:N_1]$ satisfying $i\neq j$ and $k\in[1:\lambda_1 T(1-\epsilon_T)]$, i.e., zero-forcing is performed using $M_1$ transmit antennas.
In order to apply such DL interference nulling,
\begin{align} \label{eq:constraint1}
M_1T-\lambda_1T(1-\epsilon_T)(N_1-1)\geq \lambda_1T(1-\epsilon_T)
\end{align}
should be satisfied.
Again, as seen in Figure~\ref{fig:scheme1_model1}, for reliable decoding at each DL user achieving one DoF for each information symbol,
\begin{align}  \label{eq:constraint2}
\lambda_1T(1-\epsilon_T)+\lambda_2T\leq T
\end{align}
should be satisfied. Similarly, for reliable decoding at the BS achieving one DoF for each information symbol,
\begin{align}  \label{eq:constraint3}
N_2\lambda_2T(1-\epsilon_T)\leq M_2T
\end{align}
should be satisfied.
Therefore, the proposed scheme is able to deliver $(N_1\lambda_1 +N_2\lambda_2)T(1-\epsilon_T)$ information symbols over $T$ time slots under the constraints \eqref{eq:constraint1} to \eqref{eq:constraint3}.
Finally, from the fact that $\epsilon_T\to 0$ as $T$ increases, its achievable sum DoF is represented by the following optimization problem:
\begin{align} \label{eq:maximization}
\max_{\substack{
            \lambda_1+\lambda_2\leq 1\\
            N_1\lambda_1\leq M_1\\
            N_2\lambda_2\leq M_2}}\{N_1\lambda_1+N_2\lambda_2\}.
\end{align}

\subsubsection{Achievable sum DoF}

\begin{figure}[t!]
\begin{center}
\includegraphics[scale=0.8]{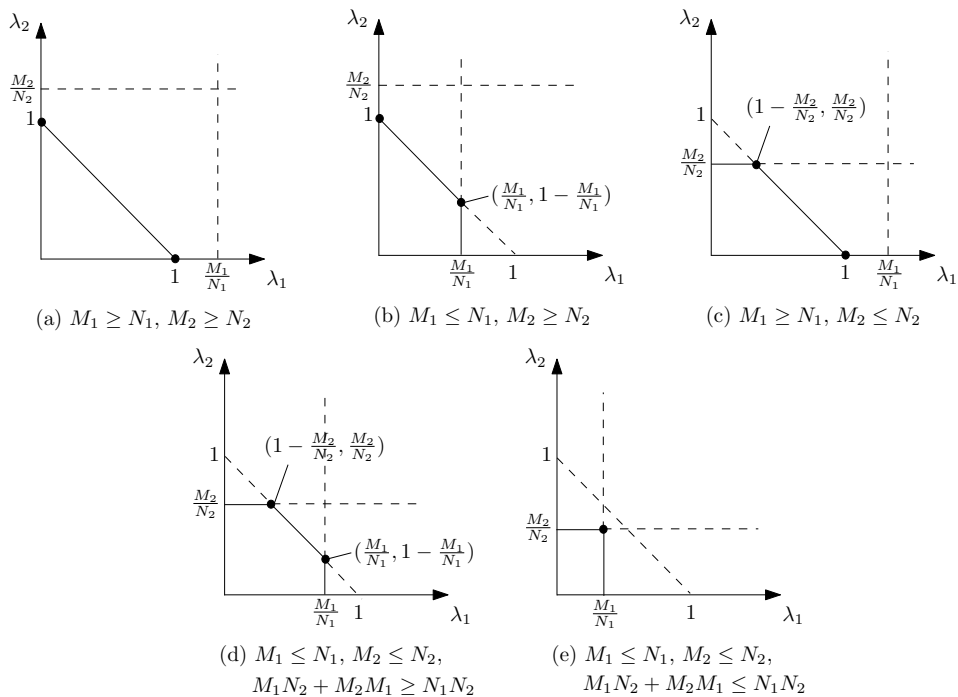}
\end{center}
\vspace{-0.15in}
\caption{Feasible $(\lambda_1,\lambda_2)$ region and the extreme points attaining  the maximum sum DoF.}
\label{fig:feasible}
\vspace{-0.1in}
\end{figure}

In the following, we prove that the sum DoF attained by solving \eqref{eq:maximization} is given as $d_{\Sigma,1}$ stated in Theorem \ref{thm:dof_1}. The linear program in \eqref{eq:maximization} is divided into five cases depending on the feasible region of $(\lambda_1,\lambda_2)$ as depicted in Figure \ref{fig:feasible}.
Obviously, one of the corner points, which are marked as points in Figure \ref{fig:feasible}, provides the maximum sum DoF.
Hence, the maximum sum DoF attained from \eqref{eq:maximization} is given by
\begin{align} \label{eq:achivable_cases}
\begin{cases}
\max(N_1,N_2)&\mbox{if }M_1\geq N_1, M_2\geq N_2,\\
\max\left(N_2,M_1+\frac{N_2(N_1-M_1)}{N_1}\right)&\mbox{if }M_1\leq N_1, M_2\geq N_2,\\
\max\left(N_1,M_2+\frac{N_1(N_2-M_2)}{N_2}\right)&\mbox{if }M_1\geq N_1, M_2\leq N_2,\\
\max\left(M_1+\frac{N_2(N_1-M_1)}{N_1},M_2+\frac{N_1(N_2-M_2)}{N_2}\right)&\mbox{if }M_1\leq N_1, M_2\leq N_2, M_1N_2+M_2N_1\geq N_1N_2,\\
M_1+M_2&\mbox{if }M_1\leq N_1, M_2\leq N_2, M_1N_2+M_2N_1\leq N_1N_2.
\end{cases}
\end{align}
The following lemma then shows that \eqref{eq:achivable_cases} is represented as $d_{\Sigma,1}$ in Theorem \ref{thm:dof_1}, which completes the achievability proof of Theorem \ref{thm:dof_1}.

\begin{lemma} \label{lemma:sum_DoF_conversion}
The sum DoF in \eqref{eq:achivable_cases} is represented as
\begin{align} \label{eq:optimal_dof}
\min\left\{M_1+M_2,\max(N_1,N_2),\max\left(M_1+\frac{N_2(N_1-M_1)}{N_1},M_2+\frac{N_1(N_2-M_2)}{N_2}\right)\right\}.
\end{align}
\end{lemma}
\begin{proof}
For notational simplicity, denote
\begin{align}
a_1&=M_1+\frac{N_2(N_1-M_1)}{N_1}=N_2+\frac{M_1(N_1-N_2)}{N_1},\nonumber\\
a_2&=M_2+\frac{N_1(N_2-M_2)}{N_2}=N_1+\frac{M_2(N_2-N_1)}{N_2}.
\end{align}
Then denote $a_3=\min\left\{M_1+M_2,\max(N_1,N_2),\max( a_1,a_2)\right\}$.
In the following, we show that for each of the five cases in \eqref{eq:achivable_cases}, $a_3$ is represented as in the corresponding DoF expression in \eqref{eq:achivable_cases}.
\begin{itemize}
\item Case I ($M_1\geq N_1, M_2\geq N_2$): Obviously, $M_1+M_2\geq \max(N_1,N_2)$. For $N_1\geq N_2$, $\max(a_1,a_2)\geq a_1\geq N_2+\frac{N_1(N_1-N_2)}{N_1}=N_1$.
For $N_1\leq N_2$, $\max(a_1,a_2)\geq a_2\geq N_1+\frac{N_2(N_2-N_1)}{N_2}=N_2$.
Hence $\max(a_1,a_2)\geq \max(N_1,N_2)$.
In conclusion, $a_3=\max(N_1,N_2)$ for Case I.
\item Case II ($M_1\leq N_1, M_2\geq N_2$):
First consider the case where $N_1\geq N_2$.
Then $M_1+M_2\geq M_1+N_2\geq M_1+\frac{N_2(N_1-M_1)}{N_1}=a_1$.
Also $\max(N_1,N_2)=N_1=N_2+\frac{N_1(N_1-N_2)}{N_1}\geq N_2+\frac{M_1(N_1-N_2)}{N_1}=a_1$.
Since $a_2\leq N_1+\frac{N_2(N_2-N_1)}{N_2}=N_2$, $\max(a_1,a_2)=a_1$.
Hence $a_3=a_1$.
Next consider the case where $N_1\leq N_2$.
Then $M_1+M_2\geq N_2$ and $\max(N_1,N_2)=N_2$.
Also $\max(a_1,a_2)\geq a_2\geq N_1+\frac{N_2(N_2-N_1)}{N_2}=N_2$.
Hence $a_3=N_2$.
Finally, from the relation that $a_1\geq N_2$ for $N_1\geq N_2$ and $a_1\leq N_2$ for $N_1\leq N_2$, $a_3=\max(N_2,a_1)$ for Case II.
\item  Case III ($M_1\geq N_1, M_2\leq N_2$): From the symmetric relation with Case II,  $a_3=\max(N_1,a_2)$ for Case III.
\item Case IV ($M_1\leq N_1, M_2\leq N_2, M_1N_2+M_2N_1\geq N_1N_2$):
The condition $M_1N_2+M_2N_1\geq N_1N_2$ means that $M_1\geq \frac{N_1(N_2-M_2)}{N_2}$ and $M_2\geq \frac{N_2(N_1-M_1)}{N_1}$. Hence $M_1+M_2\geq \frac{N_1(N_2-M_2)}{N_2}+M_2=a_2$ and $M_1+M_2\geq M_1+ \frac{N_2(N_1-M_1)}{N_1}=a_1$, which show $M_1+M_2\geq \max(a_1,a_2)$.
For $N_1\geq N_2$, $N_1= N_2+\frac{N_1(N_1-N_2)}{N_1}\geq N_2+\frac{M_1(N_1-N_2)}{N_1}=a_1$ and $N_1\geq a_2$.
Similarly, $N_2\geq a_2$ and $N_2\geq a_1$ for $N_1\leq N_2$.
Hence $\max(N_1,N_2)\geq \max(a_1,a_2)$.
In conclusion, $a_3=\max(a_1,a_2)$ for Case IV.
\item Case V ($M_1\leq N_1, M_2\leq N_2, M_1N_2+M_2N_1\leq N_1N_2$):
For $N_1\geq N_2$, $N_2(M_1+M_2)\leq M_1N_2+M_2N_1\leq N_1N_2$ and then $M_1+M_2\leq N_1$.
Similarly, $M_1+M_2\leq N_2$ for  $N_1\leq N_2$.
Hence $\max(N_1,N_2)\geq M_1+M_2$.
The condition $M_1N_2+M_2N_1\leq N_1N_2$ means that $N_1\geq \frac{N_2M_1}{N_2-M_2}$ and $N_2\geq \frac{N_1M_2}{N_1-M_1}$.
Then $a_1=M_1+\frac{N_2(N_1-M_1)}{N_1}\geq M_1+\frac{N_1M_2}{N_1-M_1}\frac{N_1-M_1}{N_1}=M_1+M_2$ and $a_2=M_2+\frac{N_1(N_2-M_2)}{N_2}\geq M_2+\frac{N_2M_1}{N_2-M_2}\frac{N_2-M_2}{N_2}=M_1+M_2$.
Hence $\max(a_1,a_2)\geq M_1+M_2$.
In conclusion, $a_3=M_1+M_2$ for Case V.
\end{itemize}

In conclusion, $a_3$ is represented as the corresponding sum DoF in \eqref{eq:achivable_cases} for all five cases, which completes the proof.
\end{proof}

\section{Converse} \label{sec:converse}
In this section, we prove the converse of Theorems \ref{thm:dof_1} and \ref{thm:dof_2}.
Recall the encoding and decoding functions of the FD BS and each FD user in Section \ref{subsec:dof_definition}.
The key observation is that the received signals available for encoding the DL messages at the FD BS and the DL messages available for decoding the UL messages at the FD BS cannot increase the sum DoF.
Similarly, the received signals available for encoding its UL message at each FD user and its UL message available for decoding its DL message at each FD user cannot increase the sum DoF.

\subsection{Converse of Theorem \ref{thm:dof_1}}

To prove the converse of Theorem \ref{thm:dof_1}, we introduce the two-user MIMO Z-IC with output feedback for encoding and message side information for decoding depicted in Figure \ref{fig:MIMO_ZIC}.
The received signal vectors of receivers 1 and 2 at time $t$ are respectively given by
\begin{align}
\mathbf{y}_1(t)=&\mathbf{H}_{11}\mathbf{x}_1(t)+\mathbf{H}_{12}\mathbf{x}_2(t)+\mathbf{z}_1(t),\nonumber\\
\mathbf{y}_2(t)=&\mathbf{H}_{22}\mathbf{x}_2(t)+\mathbf{z}_2(t),
\end{align}
where $\mathbf{H}_{11}\in \mathbb{R}^{N_1\times M_1}$, $\mathbf{H}_{12}\in \mathbb{R}^{N_1\times N_2}$, and $\mathbf{H}_{22}\in \mathbb{R}^{M_2\times N_2}$ denote the channel matrices from transmitter 1 to receiver 1, from transmitter 2 to receiver 1, and from transmitter 2 to receiver 2, respectively. The rest of the assumptions are the same as those of the FD-BS--HD-user cellular network in Section \ref{sebsec:FD_BS_HD_user}.
Obviously, the capacity of the two-user MIMO Z-IC is an outer bound on the capacity of the $(M_1,M_2,N_1,N_2)$ FD-BS--HD-user cellular network, since it corresponds to the FD-BS--HD-user cellular network with full cooperation among the DL users and among the UL users.

\begin{figure}[t!]
\begin{center}
\includegraphics[scale=0.8]{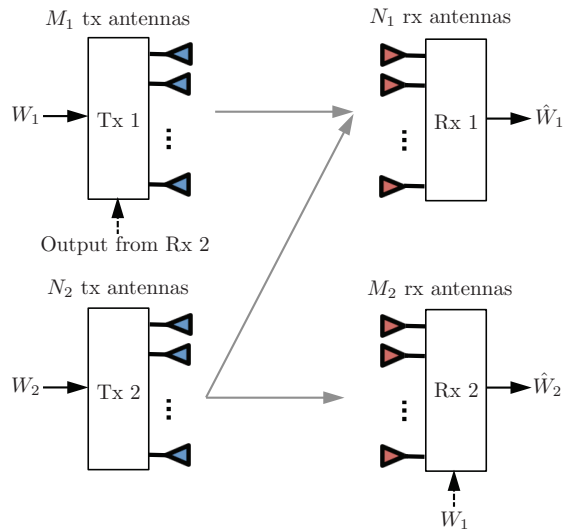}
\end{center}
\vspace{-0.15in}
\caption{Two-user MIMO Z-IC with output feedback for encoding and message side information for decoding.}
\label{fig:MIMO_ZIC}
\vspace{-0.1in}
\end{figure}

\begin{lemma}\label{lemma:MIMO_ZIC}
Consider the two-user MIMO Z-IC with output feedback for encoding and message side information for decoding in Figure \ref{fig:MIMO_ZIC}. Then the DoF region is given by the set of all DoF pairs ($d_1, d_2$) satisfying
\begin{align}
d_i&\leq \min(M_i,N_i),{~~} i=1,2\label{eq:upper1}\\
d_1+d_2&\leq \max(N_1,N_2).\label{eq:upper2}
\end{align}
\end{lemma}
\begin{proof}
The achievability immediately follows from that in \cite[Theorem 1]{Ke:12}, which corresponds to the two-user MIMO Z-IC without output feedback for encoding and message side information for decoding.
Next, we show the converse. Obviously $d_1\leq \min(M_1,N_1)$ and also $d_2\leq \min(M_2,N_2)$ since side information of $W_1$ at receiver 2 cannot increase the DoF more than $\min(M_2,N_2)$, which gives \eqref{eq:upper1}.
Now substitute $N_1$ antennas with $\max(N_1,N_2)$ antennas at receiver 1.
Assume that both receivers are able to recover $W_1$ and $W_2$ respectively with arbitrarily small probabilities of error.
Then, after subtracting $\mathbf{x}_1$ from $\mathbf{y}_1$ ($\mathbf{x}_1$ is obtained from re-encoding $W_1$), receiver 1 constructs $\mathbf{y}'_1=\mathbf{H}'_{12}\mathbf{x}_2+\mathbf{z}_1$, where $\mathbf{H}'_{12}\in\mathbb{R}^{\max(N_1,N_2)\times N_2}$.
Since receiver 2 recovers $W_2$ from $\mathbf{y}_2=\mathbf{H}_{22}\mathbf{x}_2+\mathbf{z}_2$, where $\mathbf{H}_{22}\in\mathbb{R}^{M_2\times N_2}$, receiver 1 can also recover $W_2$ from $\mathbf{y}_1'$ from the fact that $\min(N_2, \max(N_1,N_2))\geq \min(N_2,M_2)$.
As a result, receiver 1 is able to decode both $W_1$ and $W_2$ with $\max(N_1,N_2)$ antennas. Because output feedback cannot increase the sum DoF of the MIMO multiple-access channel (MAC), $d_1+d_2\leq \max(N_1,N_2)$, which provides \eqref{eq:upper2}.
In conclusion, Lemma \ref{lemma:MIMO_ZIC} holds.
\end{proof}

Since the sum DoF of the $(M_1,M_2,N_1,N_2)$ FD-BS--HD-user cellular network is upper bounded by the sum DoF of the two-user MIMO Z-IC, $d_{\Sigma,1}\leq \min(M_1+M_2, \max(N_1,N_2))$ from Lemma \ref{lemma:MIMO_ZIC}, which is yet not enough to show the converse.
In a more refined way of applying Lemma \ref{lemma:MIMO_ZIC}, we prove the converse of Theorem \ref{thm:dof_1} in the following.

%Let $d^{[\sf d]}_i$, $i\in[1:N_1]$, denote an achievable DoF of DL user $i$ and $d^{[\sf u]}_j$, $j\in[1:N_2]$, denote an achievable DoF of UL user $j$.
Denote $d^{[\sf d]}_{\Sigma}=\sum_{i=1}^{N_1}d^{[\sf d]}_i$ and $d^{[\sf u]}_{\Sigma}=\sum_{j=1}^{N_2}d^{[\sf u]}_j$.
First consider the case where $N_1\geq N_2$.
For this case, choose a subset of DL users in $\mathcal{A}^{[\sf d]}\in[1:N_1]$ satisfying $\operatorname{card}(\mathcal{A}^{[\sf d]})=N_2$.
Then, by applying Lemma \ref{lemma:MIMO_ZIC} only for the DL users in $\mathcal{A}^{[\sf d]}$ (and for the entire UL users), we have
\begin{align} \label{eq:upper_refined1}
\sum_{i\in\mathcal{A}^{[\sf d]}}d^{[\sf d]}_i+d^{[\sf u]}_{\Sigma}\leq N_2.
\end{align}
By summing \eqref{eq:upper_refined1} over all possible $\mathcal{A}^{[\sf d]}$ satisfying $\operatorname{card}(\mathcal{A}^{[\sf d]})=N_2$, we have
\begin{align} \label{eq:upper_refined2}
N_2d^{[\sf d]}_{\Sigma}+N_1d^{[\sf u]}_{\Sigma}\leq N_1N_2.
\end{align}
Therefore,
\begin{align} \label{eq:maximization_upper1}
d_{\Sigma,1}\leq\max_{\substack{
            d^{[\sf d]}_{\Sigma}\leq \min(M_1,N_1)\\
            d^{[\sf u]}_{\Sigma}\leq \min(M_2,N_2)\\
            N_2d^{[\sf d]}_{\Sigma}+N_1d^{[\sf u]}_{\Sigma}\leq N_1N_2}}\{d^{[\sf d]}_{\Sigma}+d^{[\sf u]}_{\Sigma}\}.
\end{align}

Now consider the case where $N_1\leq N_2$.
For this case, choose a subset of UL users in $\mathcal{A}^{[\sf u]}\in[1:N_2]$ satisfying $\operatorname{card}(\mathcal{A}^{[\sf u]})=N_1$.
Then applying Lemma \ref{lemma:MIMO_ZIC} for all possible $\mathcal{A}^{[\sf u]}$ satisfying $\operatorname{card}(\mathcal{A}^{[\sf u]})=N_1$ and summing them provides the same upper bound in \eqref{eq:upper_refined2}.
As a result, \eqref{eq:maximization_upper1} also holds for $N_1\leq N_2$.

By solving the linear program \eqref{eq:maximization_upper1} in a similar manner as in Section \ref{sec:achievability}, we have
\begin{align} \label{eq:upper_cases}
d_{\Sigma,1}\leq\begin{cases}
\max(N_1,N_2)&\mbox{if }M_1\geq N_1, M_2\geq N_2,\\
\max\left(N_2,M_1+\frac{N_2(N_1-M_1)}{N_1}\right)&\mbox{if }M_1\leq N_1, M_2\geq N_2,\\
\max\left(N_1,M_2+\frac{N_1(N_2-M_2)}{N_2}\right)&\mbox{if }M_1\geq N_1, M_2\leq N_2,\\
\max\left(M_1+\frac{N_2(N_1-M_1)}{N_1},M_2+\frac{N_1(N_2-M_2)}{N_2}\right)&\mbox{if }M_1\leq N_1, M_2\leq N_2, M_1N_2+M_2N_1\geq N_1N_2,\\
M_1+M_2&\mbox{if }M_1\leq N_1, M_2\leq N_2, M_1N_2+M_2N_1\leq N_1N_2.
\end{cases}
\end{align}
Note that the upper bound in \eqref{eq:upper_cases} is exactly the same as in \eqref{eq:achivable_cases}.
Therefore, from Lemma \ref{lemma:sum_DoF_conversion},
\begin{align}
d_{\Sigma,1}\leq\min\left\{M_1+M_2,\max(N_1,N_2),\max\left(M_1+\frac{N_2(N_1-M_1)}{N_1},M_2+\frac{N_1(N_2-M_2)}{N_2}\right)\right\},
\end{align}
which completes the converse proof of Theorem \ref{thm:dof_1}.

\subsection{Converse of Theorem \ref{thm:dof_2}}

In this subsection, we prove the converse of Theorem \ref{thm:dof_2}.
We first show that $d_{\Sigma,2}\leq{M_1+M_2}$ in Section \ref{subsec:two_way_bound} and then show that $d_{\Sigma,2}\leq N$ in Section \ref{subsec:x_network_bound}.
Combining the above two bounds, we have the desired bound $d_{\Sigma,2}\leq \min(M_1+M_2,N)$, which completes the converse proof.

\subsubsection{MIMO two-way network upper bound} \label{subsec:two_way_bound}
By allowing full cooperation among the $N$ users in the $(M_1,M_2,N)$ FD-BS--FD-user cellular network, we obtain a MIMO two-way network depicted in Figure~\ref{Upperbound_Twoway}. Clearly, the considered MIMO two-way network provides an upper bound on $d_{\Sigma,2}$.
Therefore, from the result in~\cite{Han:84}, we have
\begin{align}
d_{\Sigma,2}&\leq\min(M_1,N)+\min(M_2,N)\nonumber\\
&\leq M_1+M_2.
\end{align}

\begin{figure}[!t]
\centering
\includegraphics[scale=0.8]{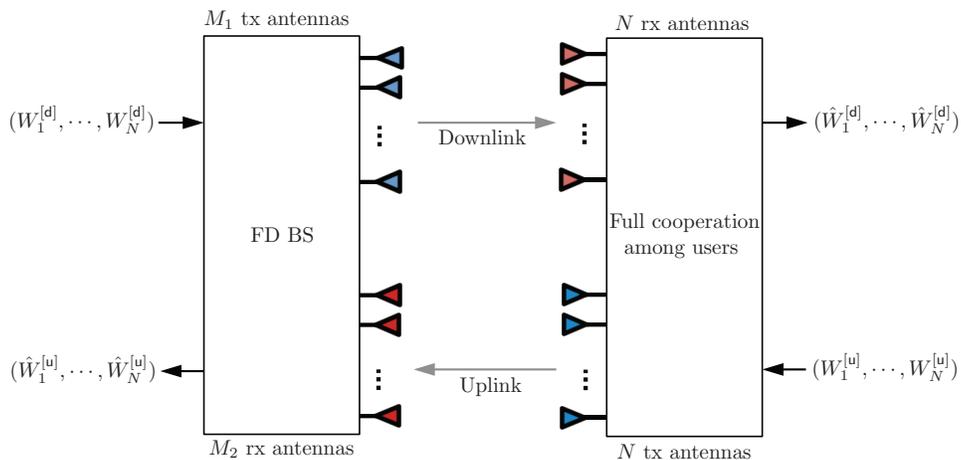}
\caption{MIMO two-way channel by allowing full cooperation among the $N$ users.} \label{Upperbound_Twoway}
\end{figure}

\subsubsection{Four-node X network upper bound} \label{subsec:x_network_bound}
We now prove $d_{\Sigma,2}\leq N$ by using the result of four-node X networks in~\cite{Viveck1:09}.
In order to apply the result in~\cite{Viveck1:09}, we convert the original $(M_1,M_2,N)$ FD-BS--FD-user cellular network into the
corresponding four-node X network as follows:

\begin{itemize}

\begin{figure}[!t]
\centering
\includegraphics[scale=0.8]{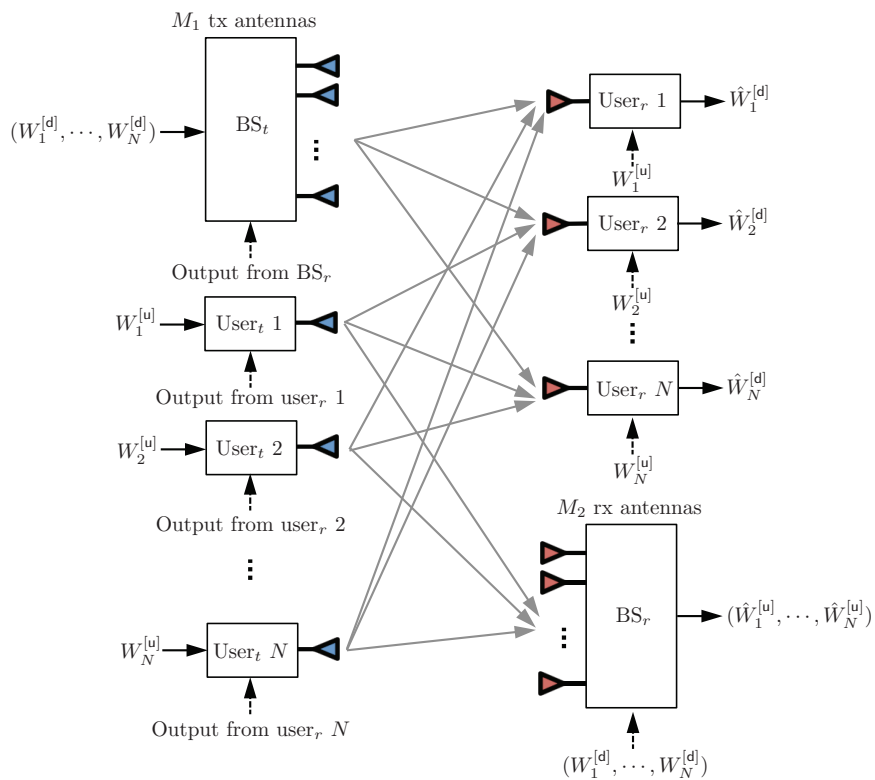}
\caption{Step 1: The equivalent two-cell network with output feedback at the encoders, and message side information at the decoders.} \label{Upperbound_xnetwork_step1}
\vspace{-0.2in}
\end{figure}

\item {\bf Step 1}:
We first transform the $(M_1, M_2, N)$ FD-BS--FD-user cellular network into the equivalent two-cell cellular network consisting of one DL cell and one UL cell depicted in Figure~\ref{Upperbound_xnetwork_step1}.
Specifically, the FD BS is decomposed into the BS$_t$ and the BS$_r$ and FD user $i$ is decomposed into user$_t~i$ and user$_r~i$, where $i\in[1:N]$.
There exists output feedback from the BS$_r$ to the BS$_t$ and from user$_r~i$ to user$_t~i$ for all $i\in[1:N]$, which can be used as side information for encoding.
In addition, $(W^{[{\sf d}]}_1\cdots,W^{[{\sf d}]}_N)$ is available at the BS$_r$ and $W^{[{\sf u}]}_i$ is available at user$_r~i$ for all $i\in[1:N]$, which can be used as side information for decoding.
We refer to the encoding and decoding functions in Section \ref{sebsec:dof_FD_BS_FD_user}.
The channel coefficients from the BS$_t$ to the BS$_r$ and from user$_t~i$ to user$_r~i$ are set to zeros due to perfect self-interference suppression in the original network.
The validity of this transformation is also proved by \cite[Lemma 1]{Viveck1:09}.

\begin{figure}[!t]
\centering
\includegraphics[scale=0.75]{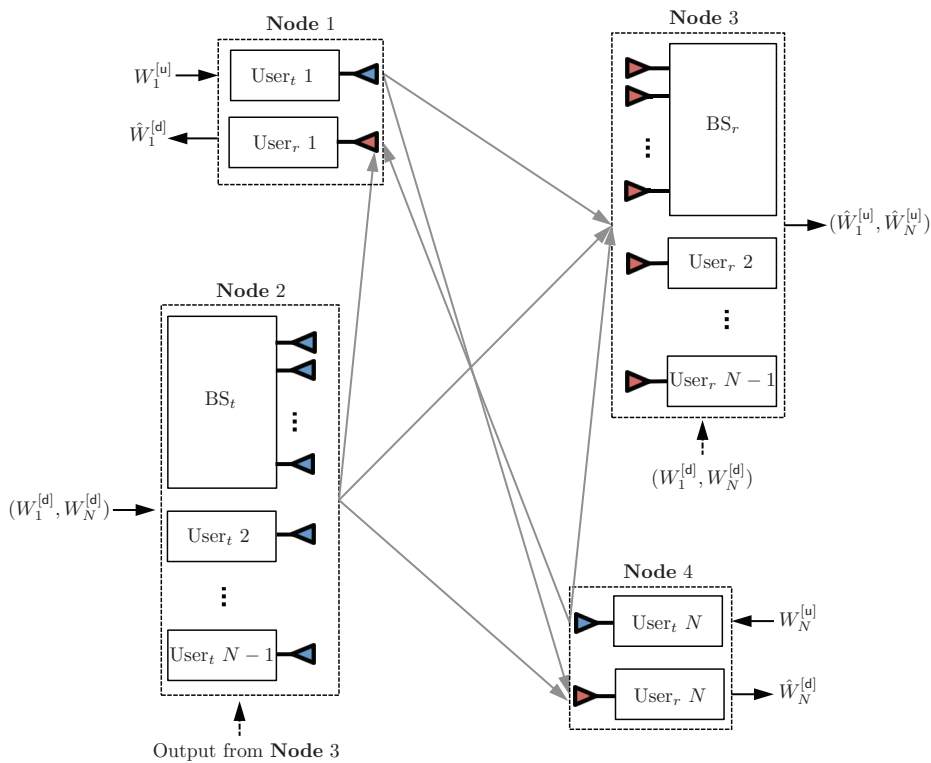}
\caption{Step 2: Cooperation between BSs and users.} \label{Upperbound_xnetwork_step2}
\end{figure}

\item {\bf Step 2}:
As shown in Figure~\ref{Upperbound_xnetwork_step2}, we allow full cooperation among user$_t~1$ and user$_r~1$, among BS$_t$, user$_t~2$ to user$_t~N-1$, among BS$_r$, user$_r~2$ to user$_r~N-1$, and among user$_t~N$ and user$_r~N$, each of which is called Nodes 1,2,3, and 4 respectively.
Because of such cooperation, the set of $(W^{[{\sf d}]}_2,\cdots,W^{[{\sf d}]}_{N-1},W^{[{\sf u}]}_2,\cdots,W^{[{\sf u}]}_{N-1})$ is priorly known at Node 3 as side information, so that Node 3 is able to attain those messages without communication.
Hence, we delete those messages in the figure without loss of generality.
In the end, Node 1 wishes to send $W_1^{[{\sf u}]}$ and estimate $W_1^{[{\sf d}]}$, Node 2 wishes to send $(W_1^{[{\sf d}]},W_N^{[{\sf d}]})$ with the help of output feedback from Node 3, i.e., the set of all output signals received by the components consisting of Node 3,  Node 3 wishes to estimate $(W_1^{[{\sf u}]},W_N^{[{\sf u}]})$ with the help of message side information $(W_1^{[{\sf d}]},W_N^{[{\sf d}]})$, and Node 4 wishes to send $W_N^{[{\sf u}]}$ and estimate $W_N^{[{\sf d}]}$.\footnote{The full cooperation assumption implies that both output feedback and message side information are available at Nodes 1 and 4.}
Since the network in Figure~\ref{Upperbound_xnetwork_step2} assumes cooperation between some nodes and allow more information for encoding and decoding, it provides an outer bound on the DoF region of the network in Figure~\ref{Upperbound_xnetwork_step1}.

\begin{figure}[!t]
\centering
\includegraphics[scale=0.75]{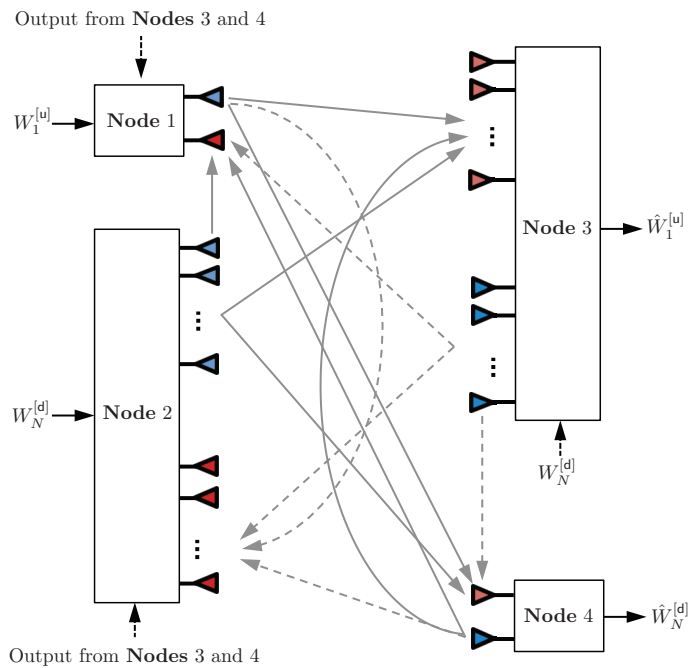}
\caption{Step 3: Eliminate all the messages except $W_1^{[\sf u]}$ and $W_N^{[\sf d]}$ and create more links and output feedback.} \label{Upperbound_xnetwork_step3}
\vspace{-0.2in}
\end{figure}

\item {\bf Step 3}: We now focus on an upper bound on $d^{[\sf u]}_1+d^{[\sf d]}_N$.
We first eliminate all the messages except $W_1^{[\sf u]}$ and $W_N^{[\sf d]}$, which does not decrease $d^{[\sf u]}_1+d^{[\sf d]}_N$~\cite{Jafar08}.
Then we provide $M_1+N-1$ receive antennas at Node 2 and $M_2+N-1$ transmit antennas at Node 3 and allow FD operation at all nodes, which creates more links illustrated as dashed links in Figure 13.
We further assume that output feedback from Nodes 3 and 4 is available at Nodes 1 and 2.
Obviously, adding more antennas at some nodes, allowing FD operation, and providing more output feedback for encoding do not decrease $d^{[\sf u]}_1+d^{[\sf d]}_N$.
\end{itemize}

As a result, the converted network in Figure \ref{Upperbound_xnetwork_step3} provides an upper bound on $d^{[\sf u]}_1+d^{[\sf d]}_N$ achievable by the original ($M_1$, $M_2$, $N$) FD-BS--FD-user cellular network.
Note that the converted network in Figure \ref{Upperbound_xnetwork_step3} corresponds to the four-node X network studied in~\cite{Viveck1:09} except the fact that $W_N^{[\sf d]}$ is provided to Node 3 through a genie.
As stated in \cite[Section IV]{Viveck1:09}, providing this side information does not increase the sum DoF and, therefore, we have $d^{[\sf u]}_1+d^{[\sf d]}_N\leq 1$ from the result in~\cite{Viveck1:09}. In the same manner, we can establish
\begin{align} \label{eq:upper_bound_di_dj}
d^{[\sf u]}_i+d^{[\sf d]}_N\leq j
\end{align}
for $i,j\in[1:N]$ with $i\neq j$.
By summing \eqref{eq:upper_bound_di_dj} for all $i,j\in[1:N]$ with $i\neq j$, we finally have
\begin{align}
d_{\Sigma,2}=\sum_{i=1}^{N}d^{[\sf u]}_{i}+\sum_{i=1}^{N}d^{[\sf d]}_{i}\leq N.
\end{align}

\section{Discussions} \label{sec:discussion}
In this section, we briefly discuss about the impacts of self-interference and UL and DL scheduling on DoF.

\begin{figure}[!t]
\centering
\includegraphics[scale=0.7]{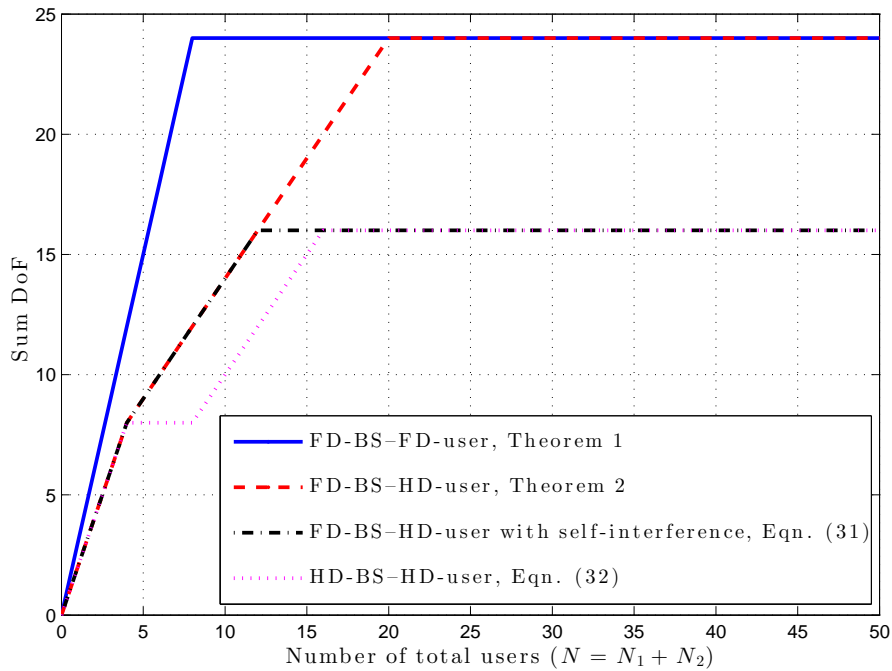}
\caption{Sum DoFs for $M_1=16$, $M_2=8$, and $N_2=2N_1$. } \label{DoF_fullduplex}
\end{figure}

\subsection{Impacts of Self-Interference on DoF}
Throughout the paper, we assumed that there is no self-interference within the BS during FD operation. However, in a practical FD BS, the amount of residual self-interference may not be negligible due to insufficient self-interference suppression or imperfect self-interference cancellation from the priorly known message information at the receiver side~\cite{Hong13}. In this subsection, we will discuss the impacts of such self-interference on the sum DoF.
%\footnote{In this paper, th e effect of self-interference within a FD mobile user is not covered.}
 Note that when there exists self-interference within the BS of the $(M_1,M_2,N_1,N_2)$ FD-BS--HD-user cellular network, the sum DoF is given by
\begin{align} \label{eq:dof_ul_dl}
\min \Bigg\{ &\frac{N_1N_2+\min(M_1,N_1)(N_1-N_2)^{+}+\min(M_2,N_2)(N_2-N_1)^+}{\max(N_1,N_2)}\nonumber,\\
   &M_1+N_2,M_2+N_1,\max(M_1,M_2),\max(N_1,N_2)\Bigg\}
\end{align}
from the result of~\cite{JeonSuh:14}, by interpreting inter-BS interference in~\cite{JeonSuh:14} as self-interference within the BS.
Obviously, if we restrict for the BS to operate either UL or DL only, then the sum DoF is given by
\begin{align} \label{eq:dof_hd_only}
\max(\min(M_1,N_1),\min(M_2,N_2)).
\end{align}

To see the effect of self-interference on the sum DoF, let us consider the case where $M_1=16$, $M_2=8$, and $N_2=2N_1$ as an example. We plot the sum DoFs as a function of the number of total users $N=N_1+N_2$ in Figure~\ref{DoF_fullduplex}.
For comparison, we also plot the sum DoF of the FD-BS--FD-user cellular network when the number of FD users is given by $N$.
As shown in the figure, the FD-BS--HD-user cellular network is able to achieve the same sum DoF attained by the FD-BS--FD-user cellular network when $N$ is large enough.
However, FD capability at the user side is beneficial to improve the sum DoF for small $N$.
Interestingly, even when there exists self-interference, FD operation at the BS alone can increase the sum DoF in a certain regime. However, the sum DoF collapses to that of the HD-BS--HD-user cellular network when $N$ is large enough. Note that similar tendencies can be observed for general $(M_1,M_2,N_1,N_2)$. Therefore, from these observations, self-interference suppression or cancellation is of crucial importance for fully utilising the potential of FD networks.

\subsection{Effects of Scheduling on DoF}
In this subsection, we discuss the effects of HD user scheduling on the sum DoF.
Suppose that there exist total $N$ HD users and we are able to coordinate the operational mode of each of these users, i.e., dividing them into $N_1$ DL users and $N_2$ UL users, where $N_1+N_2=N$.
Obviously, the sum DoF varies with the values of $N_1$ and $N_2$ from Theorem \ref{thm:dof_1}.

\begin{figure}[!t]
\centering
\includegraphics[scale=0.7]{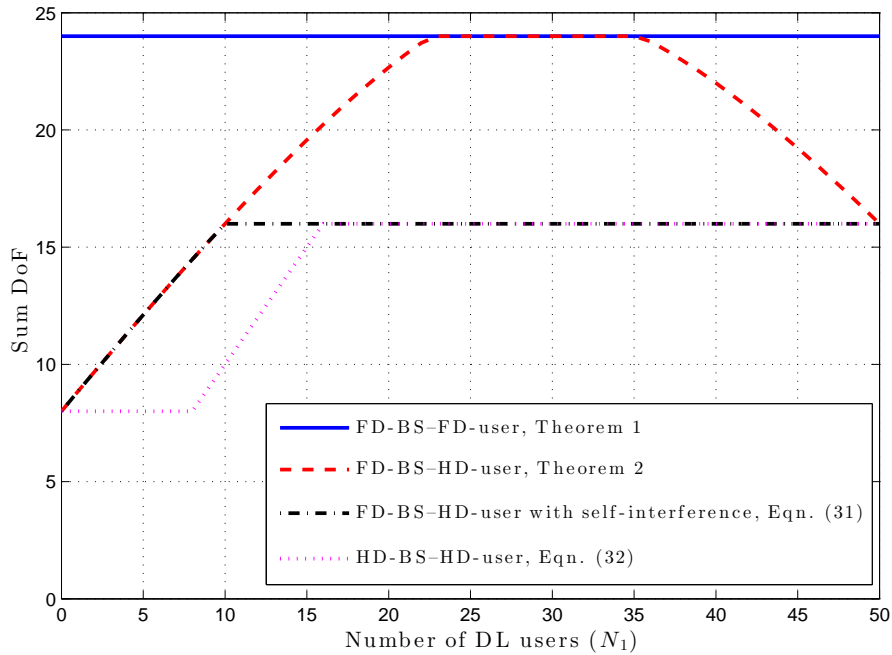}
\caption{Sum DoFs for $M_1=16$, $M_2=8$, and $N_2=50-N_1$. } \label{DoF_uplink}
\end{figure}

\begin{figure}[!t]
\centering
\includegraphics[scale=0.7]{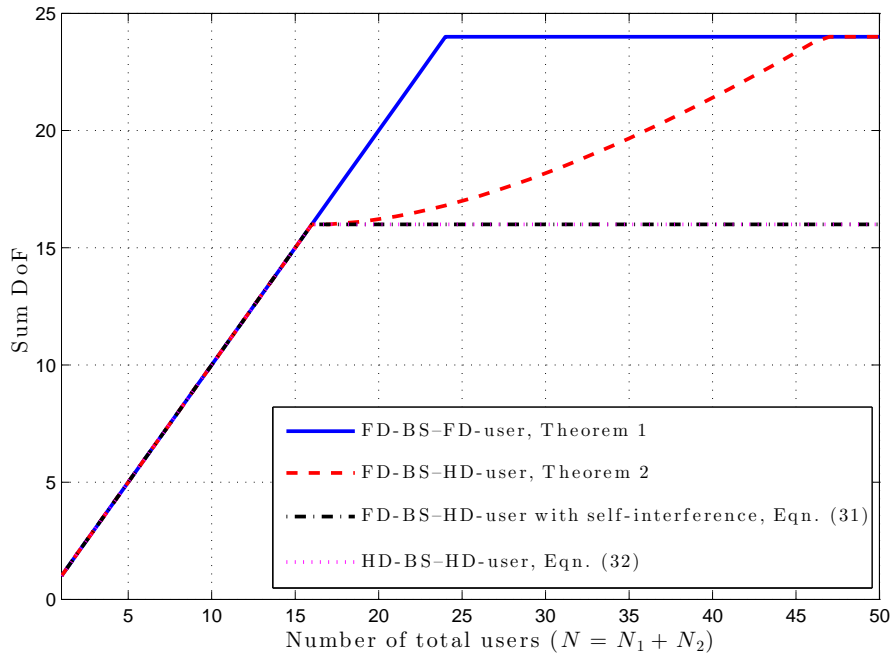}
\caption{Optimal sum DoFs for $M_1=16$ and $M_2=8$. } \label{DoF_uplink2}
\end{figure}

As an example, consider again the case where $M_1=16$ and $M_2=8$. First, we fix the total number of users $N$ ($=N_1+N_2$) as 50 and plot the sum DoF of the FD-BS--HD-user cellular network with and without self-interference suppression as a function of $N_1$ in Figure~\ref{DoF_uplink}. For comparison, we also plot the sum DoFs of the FD-BS--FD-user cellular network and the HD-BS--HD-user cellular network. As depicted in Figure~\ref{DoF_uplink}, except the FD-BS--FD-user cellular network, the achievable sum DoFs vary with $N_1$, and we can maximize the sum DoF of each network by optimally choosing $N_1$ and $N_2$.

Now, we plot the sum DoFs as a function of the number of total users $N$ in Figure~\ref{DoF_uplink2}. Here, for each $N$, we choose $N_1$ and $N_2$ to achieve the optimal sum DoFs. As seen in Figure~\ref{DoF_uplink2}, when there is no self-interference, the optimal sum DoF of the FD-BS--HD-user cellular network approaches to that of the FD-BS--FD-user cellular network and reaches the same sum DoF when $N$ is large enough. However, when there exists self-interference, the optimal sum DoFs of the FD-BS--HD-user cellular network and the HD-BS--HD-user cellular network are the same for any $N$.
This statement is also true for general $M_1$ and $M_2$ since the optimal scheduling for the FD-BS--HD-user cellular network with self-interference is to operate all HD users as either UL or DL, which can be easily verified from  \eqref{eq:dof_ul_dl}.
Therefore, for the case in which the optimal scheduling is allowed,  FD operation at the BS is not required in terms of DoF if there exists self-interference.

\section{Conclusion}
In this paper, we have studied the sum DoFs of cellular networks with a multiantenna FD BS and HD mobile users and with a multiantenna FD BS and FD mobile users. For our main contribution, we have completely characterized the sum DoFs of these networks. To be specific, for achievability, the key idea was to fully utilize the intended signal dimensions by minimizing the inter-user interference dimensions via IA for the UL transmission and by minimizing the intra-cell interference dimensions via multiantenna nulling for the DL transmission.
For converse, we have provided a matching upper bound that shows the optimality of the proposed scheme.
As a consequence of the result, we have shown that even when inter-user interference exists, FD operation at the BS can double the sum DoF over the HD only networks when the number of users becomes large enough as compared to the number of antennas at the BS,  for both the FD-BS--HD-user cellular network and the FD-BS--FD-user cellular network.

Our work can be extended to several interesting directions: (1)
Extending to multi-cell scenarios in which inter-cell interference exists; (2) Extending to the case in which mobile users have multiple antennas; (3) Extending to the
cases in which channel state information at
transmitters (CSIT) is not available or delayed.



\begin{thebibliography}{10}
\providecommand{\url}[1]{#1}
\csname url@samestyle\endcsname
\providecommand{\newblock}{\relax}
\providecommand{\bibinfo}[2]{#2}
\providecommand{\BIBentrySTDinterwordspacing}{\spaceskip=0pt\relax}
\providecommand{\BIBentryALTinterwordstretchfactor}{4}
\providecommand{\BIBentryALTinterwordspacing}{\spaceskip=\fontdimen2\font plus
\BIBentryALTinterwordstretchfactor\fontdimen3\font minus
  \fontdimen4\font\relax}
\providecommand{\BIBforeignlanguage}[2]{{%
\expandafter\ifx\csname l@#1\endcsname\relax
\typeout{** WARNING: IEEEtran.bst: No hyphenation pattern has been}%
\typeout{** loaded for the language `#1'. Using the pattern for}%
\typeout{** the default language instead.}%
\else
\language=\csname l@#1\endcsname
\fi
#2}}
\providecommand{\BIBdecl}{\relax}
\BIBdecl

\bibitem{Choi10}
J.~I. Choi, M.~Jain, K.~Srinivasan, P.~Levis, and S.~Katti, ``Achieving single
  channel, full duplex wireless communication,'' in \emph{Proc. 16th Annual
  International Conference on Mobile Computing, Networking, and Communications
  (MobiCom)}, New York, NY, Aug. 2010.

\bibitem{Aryafar12}
E.~Aryafar, M.~A. Khojastepour, K.~Sundaresan, S.~Rangarajan, and M.~Chiang,
  ``{MIDU}: Enabling {MIMO} full duplex,'' in \emph{Proc. 18th Annual
  International Conference on Mobile Computing, Networking, and Communications
  (MobiCom)}, Istanbul, Turkey, Aug. 2012.

\bibitem{Khandani13}
A.~K. Khandani, ``Two-way (true full-duplex) wireless,'' in \emph{Proc. 13th
  Canadian Workshop in Information Theory (CWIT)}, Toronto, Canada, Jun. 2013.

\bibitem{Duarte10}
M.~Duarte and A.~Sabharwal, ``Full-duplex wireless communications using
  off-the-shelf radios: Feasibility and first results,'' in \emph{Proc.
  Asilomar Conference on Signals, Systems and Computers}, Pacific Grove, CA,
  Nov. 2010.

\bibitem{Jain11}
M.~Jainy, J.~I. Choi, T.~M. Kim, D.~Bharadia, S.~Seth, K.~Srinivasan, P.~Levis,
  S.~Katti, and P.~Sinha, ``Practical, real-time, full duplex wireless,'' in
  \emph{Proc. 17th Annual International Conference on Mobile Computing,
  Networking, and Communications (MobiCom)}, Las Vegas, NV, Sep. 2011.

\bibitem{Bharadia13}
D.~Bharadia, E.~Mcmilin, and S.~Katti, ``Full duplex radios,'' in \emph{Proc.
  ACM SIGCOMM}, Hong Kong, China, Aug. 2013.

\bibitem{Hong13}
S.~Hong, J.~Brand, J.~Choi, M.~Jain, J.~Mehlman, S.~Katti, and P.~Levis,
  ``Applications of self-interference cancellation in 5{G} and beyond,''
  \emph{{IEEE} Communications Magazine}, vol.~52, pp. 114--121, Feb. 2014.

\bibitem{Jafar08}
S.~A. Jafar and S.~{Shamai (Shitz)}, ``Degrees of freedom region for the {MIMO}
  {X} channel,'' \emph{{IEEE} Trans. Inf. Theory}, vol.~54, pp. 151--170, Jan.
  2008.

\bibitem{Maddah-Ali:08}
M.~A. Maddah-Ali, A.~S. Motahari, and A.~K. Khandani, ``Communication over
  {MIMO} {X} channels: {I}nterference alignment, decomposition, and performance
  analysis,'' \emph{{IEEE} Trans. Inf. Theory}, vol.~54, pp. 3457--3470, Aug.
  2008.

\bibitem{Cadambe107}
V.~R. Cadambe and S.~A. Jafar, ``Interference alignment and degrees of freedom
  for the {$K$}-user interference channel,'' \emph{{IEEE} Trans. Inf. Theory},
  vol.~54, pp. 3425--3441, Aug. 2008.

\bibitem{Suh08}
C.~H. Suh and D.~Tse, ``Interference alignment for cellular networks,'' in
  \emph{Proc. 46th Annu. Allerton Conf. Communication, Control, and Computing},
  Monticello, IL, Sep. 2008.

\bibitem{Suh:11}
C.~Suh and D.~Tse, ``Downlink interference alignment,'' \emph{{IEEE} Trans.
  Commun.}, vol.~59, pp. 2616--2626, Sep. 2011.

\bibitem{Viveck1:09}
V.~R. Cadambe and S.~A. Jafar, ``Degrees of freedom of wireless networks with
  relays, feedback, cooperation, and full duplex operation,'' \emph{{IEEE}
  Trans. Inf. Theory}, vol.~55, pp. 2334--2344, May 2009.

\bibitem{Viveck2:09}
------, ``Interference alignment and the degrees of freedom of wireless {$X$}
  networks,'' \emph{{IEEE} Trans. Inf. Theory}, vol.~55, pp. 3893--3908, Sep.
  2009.

\bibitem{Tiangao:10}
T.~Gou and S.~A. Jafar, ``Degrees of freedom of the {$K$} user {$M\times N$}
  {MIMO} interference channel,'' \emph{{IEEE} Trans. Inf. Theory}, vol.~56, pp.
  6040--6057, Dec. 2010.

\bibitem{Annapureddy:11}
V.~S. Annapureddy, A.~{El Gamal}, and V.~V. Veeravalli, ``Degrees of freedom of
  interference channels with {CoMP} transmission and reception,'' \emph{{IEEE}
  Trans. Inf. Theory}, vol.~58, pp. 5740--5760, Sep. 2012.

\bibitem{Ke:12}
L.~Ke, A.~Ramamoorthy, Z.~Wang, and H.~Yin, ``Degrees of freedom region for an
  interference network with general message demands,'' \emph{{IEEE} Trans. Inf.
  Theory}, vol.~58, pp. 3787--3797, Jun. 2012.

\bibitem{Tiangao:12}
T.~Gou, S.~A. Jafar, C.~Wang, S.-W. Jeon, and S.-Y. Chung, ``Aligned
  interference neutralization and the degrees of freedom of the
  $2\times2\times2$ interference channel,'' \emph{{IEEE} Trans. Inf. Theory},
  vol.~58, pp. 4381--4395, Jul. 2012.

\bibitem{Jeon4:12}
S.-W. Jeon and M.~Gastpar, ``A survey on interference networks: Interference
  alignment and neutralization,'' \emph{Entropy}, vol.~14, pp. 1842--1863, Sep.
  2012.

\bibitem{JeonSuh:14}
S.-W. Jeon and C.~Suh, ``Degrees of freedom of uplink--downlink multiantenna
  cellular networks,'' in \emph{arXiv:cs.IT/1404.0612,} Apr. 2014.

\bibitem{Sahai13}
A.~Sahai, S.~Diggavi, and A.~Sabharwal, ``On degrees-of-freedom of full-duplex
  uplink/downlink channel,'' in \emph{Proc. {IEEE} Information Theory Workshop
  (ITW)}, Sevilla, Spain, Sep. 2013.

\bibitem{Nazer11:09}
B.~Nazer, M.~Gastpar, S.~A. Jafar, and S.~Vishwanath, ``Ergodic interference
  alignment,'' \emph{{IEEE} Trans. Inf. Theory}, vol.~58, pp. 6355--6371, Oct.
  2012.

\bibitem{Motahari:09}
A.~Motahari, S.~O. Gharan, and A.~Khandani, ``Real interference alignment with
  real numbers,'' [Online]. Available: http://arxiv.org/abs/0908.1208, Aug.
  2009.

\bibitem{Motahari:091}
------, ``Real interference alignment: Exploiting the potential of single
  antenna systems,'' \emph{{IEEE} Trans. Inf. Theory}, vol.~60, Aug 2014.

\bibitem{Jeon5:13}
S.-W. Jeon and S.-Y. Chung, ``Capacity of a class of linear binary field
  multisource relay networks,'' \emph{{IEEE} Trans. Inf. Theory}, vol.~59, pp.
  6405--6420, Oct. 2013.

\bibitem{Jeon2:11}
S.-W. Jeon, S.-Y. Chung, and S.~A. Jafar, ``Degrees of freedom region of a
  class of multisource {G}aussian relay networks,'' \emph{{IEEE} Trans. Inf.
  Theory}, vol.~57, pp. 3032--3044, May 2011.

\bibitem{Jeon2:14}
S.-W. Jeon, C.-Y. Wang, and M.~Gastpar, ``Approximate ergodic capacity of a
  class of fading two-user two-hop networks,'' \emph{{IEEE} Trans. Inf.
  Theory}, vol.~60, pp. 866--880, Feb. 2014.

\bibitem{Kim:11}
T.~Kim, D.~J. Love, and B.~Clerckx, ``On the spatial degrees of freedom of
  multicell and multiuser {MIMO} channels,'' in \emph{arXiv:1111.3160}, Nov.
  2011.

\bibitem{Shin:11}
W.~Shin, N.~Lee, J.-B. Kim, C.~Shin, and K.~Jang, ``On the design of
  interference alignment scheme for two-cell {MIMO} interfering broadcast
  channels,'' \emph{{IEEE} Trans. Wireless Commun.}, vol.~10, pp. 437--442,
  Feb. 2011.

\bibitem{Liu2:13}
T.~Liu and C.~Yang, ``Genie chain and degrees of freedom of symmetric {MIMO}
  interference broadcast channels,'' in \emph{arXiv:cs.IT/1309.6727}, Sep.
  2013.

\bibitem{Liu:13}
------, ``On the feasibility of linear interference alignment for {MIMO}
  interference broadcast channels with constant coefficients,'' \emph{{IEEE}
  Trans. Signal Processing}, vol.~61, pp. 2178--2191, May 2013.

\bibitem{Sridharan:13}
G.~Sridharan and W.~Yu, ``Degrees of freedom of {MIMO} cellular networks:
  {D}ecomposition and linear beamforming design,'' in
  \emph{arXiv:cs.IT/1312.2681}, Dec. 2013.

\bibitem{Park:12}
S.-H. Park and I.~Lee, ``Degrees of freedom for multually interfering broadcast
  channels,'' \emph{{IEEE} Trans. Inf. Theory}, vol.~58, pp. 393--402, Jan.
  2012.

\bibitem{KimJeon:14}
K.~Kim, S.-W. Jeon, J.~Yang, and D.~K. Kim, ``The feasibility of interference
  alignment for reverse {TDD} systems in {MIMO} cellular networks,'' in
  \emph{arXiv:cs.IT/1410.4624}, Oct. 2014.

\bibitem{Han:84}
T.~S. Han, ``A general coding scheme for the two-way channel,'' \emph{{IEEE}
  Trans. Inf. Theory}, vol.~30, no.~1, pp. 35--44, Jan. 1984.

\end{thebibliography}
\end{document}